\documentclass[sigplan,screen]{acmart}

\setcopyright{acmlicensed}
\acmPrice{15.00}
\acmDOI{10.1145/3315507.3330198}
\acmYear{2019}
\copyrightyear{2019}
\acmISBN{978-1-4503-6718-9/19/06}
\acmConference[DBPL '19]{Proceedings of the 17th ACM SIGPLAN International Symposium on Database Programming Languages}{June 23, 2019}{Phoenix, AZ, USA}
\acmBooktitle{Proceedings of the 17th ACM SIGPLAN International Symposium on Database Programming Languages (DBPL '19), June 23, 2019, Phoenix, AZ, USA}

\usepackage{booktabs}   
\usepackage{subcaption} 
\usepackage{hyphenat} 
\usepackage[utf8]{inputenc}
\usepackage{listings}
\usepackage{mathpartir}
\usepackage{newunicodechar}
\usepackage{xspace}
\usepackage{balance}

\newtheorem{theorem}[figure]{Theorem}
\newtheorem{remark}[figure]{Remark}
\newtheorem{lemma}[figure]{Lemma}
\theoremstyle{definition}
\newtheorem{definition}[figure]{Definition}

\newcommand\rec[1]{\ensuremath{\langle}#1\ensuremath{\rangle}}
\newcommand\coll[1]{\texttt{[}#1\texttt{]}}

\newcommand\intty[0]{\texttt{Int}}
\newcommand\inttyc[0]{\texttt{Int}^*}
\newcommand\boolty[0]{\texttt{Bool}}

\newcommand{\Haskell}{\textsc{Haskell}\xspace}
\newcommand{\TLinks}{\textsc{Links\textsuperscript{T}}\xspace}

\newcommand{\NRC}{\textsc{Nrc}\xspace}
\newcommand{\LLinks}{\textsc{Links\textsuperscript{L}}\xspace}
\newcommand{\Links}{\textsc{Links}\xspace}

\newcommand{\SQL}{\textsc{Sql}\xspace}
\newcommand{\DSH}{\textsc{Dsh}\xspace}

\newcommand{\lamiml}{$\lambda_i^\mathit{ML}$\xspace}
\newcommand{\lamr}{$\lambda_r$\xspace}
\newcommand{\kw}[1]{\text{\tt\bfseries#1}}
\newcommand{\concat}{\ensuremath{+\!\!\!\!+\,}}
\newcommand{\Urweb}{\textsc{Ur/Web}\xspace}

\newcommand{\ttrace}[2]{\llbracket #1 : \mathsf{T}(#2)\rrbracket}
\newcommand{\trace}[1]{\llbracket #1 \rrbracket}

\newunicodechar{⟦}{\ensuremath{\llbracket}}
\newunicodechar{⟧}{\ensuremath{\rrbracket}}
\newunicodechar{↦}{\mapsto}
\newunicodechar{⋅}{\ensuremath{\cdot}}
\newunicodechar{𝔇}{\ensuremath{\mathfrak{D}}}
\newunicodechar{𝔏}{\ensuremath{\mathfrak{L}}}
\newunicodechar{𝔚}{\ensuremath{\mathfrak{W}}}
\newunicodechar{ℰ}{\ensuremath{\mathcal{E}}}

\newunicodechar{Γ}{\ensuremath{\Gamma}}
\newunicodechar{⊢}{\ensuremath{\vdash}}
\newunicodechar{τ}{\ensuremath{\tau}}
\newunicodechar{⟶}{\ensuremath{\longrightarrow}}
\newunicodechar{∧}{\ensuremath{\wedge}}
\newunicodechar{∀}{\ensuremath{\forall}}
\newunicodechar{λ}{\ensuremath{\lambda}}
\newunicodechar{α}{\ensuremath{\alpha}}
\newunicodechar{✓}{\checkmark}
\newunicodechar{⤳}{\ensuremath{\leadsto}}

\definecolor{keyword}{RGB}{96,0,53}
\definecolor{darkgreen}{RGB}{0,128,0}
\definecolor{listingsFrame}{gray}{0.6}
\definecolor{string}{RGB}{0,128,0}
\definecolor{string}{RGB}{0,0,128}

\makeatletter
\lstdefinestyle{normal}{%
  aboveskip=.5\baselineskip,
  belowskip=.5\baselineskip,
  basicstyle=%
    \tt%
    \lst@ifdisplaystyle\footnotesize\fi
}
\makeatother

\lstdefinestyle{infigure}{%
  aboveskip=0em,
  belowskip=0em,
  xleftmargin=0em,
  xrightmargin=0em,
  basicstyle=\footnotesize\tt,
  stringstyle=\footnotesize\it,
}

\newcommand{\mytilde}{%
  \texttt{\resizebox{.48em}{1ex}{\hbox{$\sim$}}}%
}

\lstdefinelanguage{Links}{%
  keywords={delete,fun,if,else,for,from,sig,server,where,query,table,with,var,typename},
  mathescape=true,
  comment=[l]{\#},
  morestring=[b]/,
  morestring=[b]",
  sensitive=true,
  escapechar=\%,
}

\lstdefinelanguage{WLinks}[]{Links}{%
  language=Links,
  morekeywords={data,prov,default}
}
\lstdefinelanguage{LLinks}[]{Links}{%
  language=Links,
  morekeywords={lineage}
}

\lstdefinelanguage{TLinks}[]{Links}{%
  language=Links,
  morekeywords={Typerec, true, false, trace, Typerec, then, Rmap, rmap, rfold, tracecase, typecase, of, fix},
  literate={{/}{{/}}{1}{\{}{{$\langle$}}{1}{\}}{{$\rangle$}}{1}{λ}{{$\lambda$}}{1}}
}

\lstdefinelanguage{LaTeX}[]{Tex}{%
  language=Tex,
  morekeywords={emph, label, ref, begin, end, part, chapter, section,
                subsection, subsubsection, paragraph, subparagraph, cite},
}

\lstdefinelanguage{SugarJ}[]{Java}{%
  language=Java,
  morekeywords={sugar, context, free, syntax, desugarings, sorts, signature,
    constructors, rules, strategies, assert, as, editor, services, colorer, folding,
    outliner, checks, recursive, errors, warnings, where,
    css, outlining, rec, color, completions, completion, template,
    layout, analyses, then, extension},
  mathescape=true,
  deletestring=[b]',
  morecomment=[l]{//},
  escapechar=\%,
}

\lstdefinelanguage{SugarJXML}[]{SugarJ}{%
  morekeywords={xmlschema},
  moredelim=*[s][\color{blue}]{<}{>},
  emph={$}, 
  emphstyle={\bfseries\color{keyword}},
  deletecomment=[l]{//},
  morecomment=[s]{<!--}{-->}
}

\lstdefinelanguage{SugarMDD}[]{SugarJ}{%
  morekeywords={model, transformation},
}

\lstdefinelanguage{SugarJATM}[]{SugarMDD}{%
  morekeywords={statemachine, initial, state, events},
}

\lstdefinelanguage{SugarJEntity}[]{SugarMDD}{%
  morekeywords={entity},
}
\lstdefinelanguage{SugarJATMEntity}[]{SugarJATM}{%
  morekeywords={data, entity},
}

\lstdefinelanguage{SugarJTemplate}[]{SugarMDD}{%
  morekeywords={in, \$for, template},
  mathescape=false
}

\lstdefinelanguage{SugarJFeature}[]{SugarMDD}{%
  morekeywords={featuremodel, features, constraint, config, enable, disable, variable},
  otherkeywords={\#ifdef},
}

\lstdefinelanguage{scala}{
  morekeywords={abstract,case,catch,class,def,%
    do,else,extends,false,final,finally,%
    for,if,implicit,import,match,mixin,%
    new,null,object,override,package,%
    private,protected,requires,return,sealed,%
    super,this,throw,trait,true,try,%
    type,val,var,while,with,yield},
  sensitive=true,
  morecomment=[l]{//},
  morecomment=[n]{/*}{*/},
  morestring=[b]",
  morestring=[b]',
  morestring=[b]"""
}

\lstdefinelanguage{SDF}{%
  morekeywords={context, free, syntax, sorts, signature,
    constructors, rules, strategies},
  escapechar=_,
  mathescape=true,
  comment=[l]{\%\%},
  morestring=[b]",
}

\lstdefinelanguage{MyPython}[]{Python}{%
  escapechar=\%,
  mathescape=true,
}

\lstdefinelanguage{MyHaskell}[]{Haskell}{%
  escapechar=\%,
  mathescape=true,
  deletekeywords={Nothing,Just,False,True,putStrLn,fail,fromJust,lookup,Num,exp,free,snd,String,
  return,error,otherwise,not,show,read,Eval,Read,readsPrec,print},
}

\lstdefinelanguage{SugarHaskell}[]{MyHaskell}{%
  morekeywords={context, free, syntax, desugarings, sorts, signature,
    constructors, rules, strategies, lexical, reject},
  mathescape=false,
}

\lstdefinelanguage{SugarHaskellArrows}[]{SugarHaskell}{%
  morekeywords={proc},
}

\lstdefinelanguage{EBNF}{
  morestring=[b]"
}

\lstdefinelanguage{Constraint}{
}

\lstdefinelanguage{Plain}{}

\lstdefinelanguage{Questionnaire}[]{Java}{%
  morekeywords = {questionnaire, question, value, Boolean, String, Integer, group,
    if, else, define, ask}
}

\lstdefinelanguage{SugarFomega}{
  keywords = {module, val, type, mu, if, then, else, case, of,
    fold, unfold, true, false, as, public, import, syntax, desugaring, typing, context, free, let, in, end, forall, do},
  mathescape = true,
  morestring=[b]",
}

\begin{document}

\lstset{language=TLinks}

\title{Language-integrated provenance by trace analysis}         


\author{Stefan Fehrenbach}
\affiliation{
  \institution{University of Edinburgh}            
  \country{United Kingdom}                    
}
\email{stefan.fehrenbach@gmail.com}          

\author{James Cheney}
\affiliation{
   \institution{University of Edinburgh and The Alan Turing Institute}            
  \country{United Kingdom}                   
}
\email{jcheney@inf.ed.ac.uk}         

\begin{abstract}
  Language-integrated provenance builds on language\hyp{}integrated query
  techniques to make provenance information explaining query results
  readily available to programmers.  In previous work we have explored
  language-integrated approaches to provenance in \Links and \Haskell.
  However, implementing a new form of provenance in a
  language\hyp{}integrated way is still a major challenge.  We propose a
  self-tracing transformation and trace analysis features that,
  together with existing techniques for type-directed generic
  programming, make it possible to define different forms of
  provenance as user code.  We present our design as an extension to a
  core language for \Links called \TLinks, give examples showing its
  capabilities, and outline its metatheory and key correctness
  properties.
\end{abstract}




\maketitle

\section{Introduction}
Provenance tracking has been heavily investigated as a means of making
database query results explainable~\cite{ICDT2001BunemanKT,FTDB2009CheneyCT}, for example to
explain where in the input some output data came from
(\emph{where-provenance}) or what input records justify the presence of
some output record (\emph{lineage}, \emph{why-provenance}).  Many
prototype implementations of provenance-tracking have been developed
as ad hoc extensions to (or middleware layers wrapping) ordinary relational database systems~\cite{Senellart:2018:PPP:3229863.3275590,GProMArab2014}, 
typically by augmenting the data model with additional annotations and
propagating them through the query using an enriched semantics.  This
approach, however, inhibits reuse and uptake of these techniques since
a special (and usually not maintained) variant of the database system
must be used.  Installing, maintaining and using such research
prototypes is not for the faint of heart.

We advocate a \emph{language-based} approach to provenance, building
on \emph{language-integrated
  query}~\cite{meijer06sigmod,TLDI2012LindleyC,SIGMOD2014CheneyLW}.
In language-integrated query, database queries are embedded in a
programming language as first-class citizens, not uninterpreted
strings, and thus benefit from typechecking and other language
services.  In language-integrated \emph{provenance}, we aim to support
provenance-tracking techniques by modifying the behavior of queries at
the language level to track their own provenance. These modified
queries can then be used with unmodified, mainstream database
systems. To date, Fehrenbach and Cheney~\cite{FEHRENBACH2018103} have demonstrated the
capabilities of language\hyp{}integrated provenance in \Links, a Web and
database programming language, and Stolarek
and Cheney~\cite{Programming2018StolarekC} adapted this approach to
work with \DSH, an existing language\hyp{}integrated query library in
\Haskell~\cite{SIGMOD2015UlrichG}.  In both cases, where-provenance
and lineage are supported as representative forms of provenance.

However, both approaches explored so far have drawbacks.  Our
previous implementations of language-integrated provenance in \Links are ad hoc
language extensions, requiring nontrivial changes to the \Links
front-end and runtime.  It is not obvious how to support both
extensions at once, and supporting additional extensions would
likewise require a major intervention to the language.  In \DSH,
we were able to support both forms of provenance at
once, but did need to make superficial changes to \DSH and carry
out nontrivial type-level
programming to make our translations pass \Haskell's typechecker.
Thus, in both cases, we feel there is significant room for
improvement, to make it easier to develop new forms of
provenance without ad hoc language extensions or
subtle type-level programming.

In this paper, we present a core language design called \TLinks that
extends the query language core of \Links (a variant of the Nested
Relational Calculus~\cite{TCS1995BunemanNTW}) with several powerful
programming constructs.  These include well-studied constructs for
type-directed generic programming (e.g. $\kw{Typerec}$ and
$\kw{typecase}$)~\cite{Harper:1995:CPU:199448.199475}, extended to support
generic programming with record types~\cite{PLDI2010Chlipala}.
In addition, we propose novel primitives for constructing and
analyzing \emph{query traces} (following~\cite{cheney14ppdp}).  We
will show that these features suffice to define forms of provenance
programmatically, using the following recipe.  Given a query $q$, we
first \emph{transform} it to a \emph{self-tracing} query $q^{\mathsf{T}}$.  We
can then \emph{compose} $q^{\mathsf{T}}$ with a \emph{trace analysis function}
$f^{\mathsf{P}}$, which is simply an ordinary \TLinks function that makes use of
the type and trace analysis capabilities.  Each form of provenance we
support can
be defined as a trace analysis function, and can be applied to queries
of any type.  Thus, $f^{\mathsf{P}} \circ q^{\mathsf{T}}$ defines the intended query result
together with the desired provenance.  Finally, we \emph{normalize}
$f^{\mathsf{P}} \circ q^{\mathsf{T}}$ to a \NRC expression, which can be further translated to \SQL and evaluated efficiently on a mainstream database by the existing language-integrated query implementation in \Links \cite{SIGMOD2014CheneyLW}.  Normalization effectively
deforests the traces that would be produced by $q^{\mathsf{T}}$ if we were to
execute it directly; thus, executing the normalized \NRC query is
typically much faster than executing $q^{\mathsf{T}}$ and then $f^{\mathsf{P}}$ separately
would be.

Our main contributions are as follows.
\begin{itemize}
\item We show via examples (Section~\ref{sec:trace-analysis}) how a
  programmer can use type and trace analysis constructs to define
  different modifications of query behavior, for example to extract
  where-provenance and lineage from traces.
\item We present the language design of \TLinks. We informally
  introduce the novel trace constructors in
  Section~\ref{sec:query-traces} and present syntax and type system
  details in Section~\ref{sec:syntax-semantics}. This includes traces
  and trace analysis operations, and reviews the already-studied
  type-directed generic programming features from previous work.
\item We then present the self-tracing transformation
  (Section~\ref{sec:self-tracing}) and the extended rewrite rule
  system needed for normalization, and outline the proofs of type
  preservation and correctness for these components
  (Section~\ref{sec:normalization}).
\end{itemize}
We have a preliminary implementation, but the main
contributions of this paper concern the design and theory, and a
full-scale implementation in \Links is future work.

\section{The problem}
  
As explained earlier, in previous work we have investigated different
ways of implementing where-provenance and lineage on top of existing
language-integrated query systems, namely \Links and \DSH.  In both
cases, given a query $q$, we wish to construct another query $q^{\mathsf{P}}$
that provides both the ordinary query results of $q$ and additional
\emph{annotations} that provide some form of information about how
query results relate to the input data.  Preferably, the transformed
query should still be in the same query language as that handled by the
existing language-integrated query system, so that this implementation
can be reused to generate efficient \SQL queries.  Of course, in a typed programming language, we also expect the generated query to be well-typed.

For example, for
where-provenance, we wish to construct query $q^\mathsf{where}$ in
which each data field in the query result is annotated with a
\emph{source location} in the input database, which we
typically implement as a tuple $(R,A,i)$ consisting of a relation name
$R$, attribute name $A$, and row identifier (or primary key value)
$i$.  Likewise, for lineage, we wish to construct a query
$q^{\mathsf{lineage}}$ in which each output record is annotated with a
collection of references $(R,i)$ to input records that help
``witness'' or ``justify'' the presence of the output record.

\begin{figure}[tb]
\small
\textbf{Agencies}\\[.3em]
\begin{tabular}{rlll}
(oid) & name & based\_in & phone\\
\cmidrule{2-4}
1 & EdinTours & Edinburgh & 412 1200\\
2 & Burns's & Glasgow & 607 3000
\end{tabular}
\\[1em]

\textbf{ExternalTours}\\[.3em]
\begin{tabular}{rlllr}
(oid) & name & destination & type & price (in \pounds)\\
\cmidrule{2-5}
3 & EdinTours & Edinburgh & bus & 20\\
4 & EdinTours & Loch Ness & bus & 50\\
5 & EdinTours & Loch Ness & boat & 200\\
6 & EdinTours & Firth of Forth & boat & 50\\
7 & Burns's & Islay & boat & 100\\
8 & Burns's & Mallaig & train & 40
\end{tabular}
\\[1em]

\textbf{BoatToursQueryResult}\\[.3em]
\begin{tabular}{ll}
  name & phone\\
  \midrule
  EdinTours & 412 1200\\
  EdinTours& 412 1200\\
  Burns's & 607 3000
\end{tabular}

\caption{Example database and boat tours query result.}\label{fig:example-data}
\end{figure}

As a running example, consider the following boat tours query (in
\Links syntax). It uses nested \lstinline|for| comprehensions to
iterate over two tables, filtering by type and joining on the name
columns. It returns a list of records (pairs of field name and value
separated by commas and enclosed in angle brackets) containing the
agencies names and phone numbers. See Figure~\ref{fig:example-data} for
an example input database and result.
\begin{lstlisting}
for (e <- externalTours) where (e.type == "boat")
  for (a <- agencies) where (a.name == e.name) 
    [{name = e.name, phone = a.phone}]
\end{lstlisting}

The where-provenance translation of
this query should annotate the field value Burns's in
the result with where-provenance annotation (ExternalTours, name, 7),
and the lineage translation should annotate the row (Burns's, 607
3000) with lineage annotation [(Agencies,2), (ExternalTours,7)]. (Note
that in lineage, the annotation of each row is a \emph{collection} of
input row references; both \Links and \DSH can already handle such nested query results~\cite{SIGMOD2014CheneyLW,SIGMOD2015UlrichG}.)

In our previous work, we have implemented these translations either by
directly changing the language implementation (in \Links), or by
making nontrivial modifications to a language-integrated query library
(in \DSH).  While this work shows that it is possible to provide
(reasonably efficient) language-integrated provenance via
source-to-source translation of queries, both approaches are still
nontrivial interventions to an existing implementation, and so
developing new forms of provenance, or variations on existing ones, is
still a considerable challenge.

If we wish to provide the necessary query transformation
capability using high-level programming constructs, then we face two
significant challenges.  First, transforming the query expression in
the direct approaches considered so far
relies on fairly heavyweight metaprogramming capabilities, and
type-safe metaprogramming by reflection over object languages with
binding constructs (such as comprehensions in queries) is a
significant challenge.  Based on prior work on
general forms of provenance such as
\emph{traces}~\cite{10.1007/978-3-642-28641-4_22,cheney14ppdp} or \emph{provenance
  polynomials}~\cite{green07pods}, we might hope to avoid the need for
heavyweight metaprogramming by computing a single, general form of
\emph{query trace} once and for all, and specializing it to different forms of
provenance later.  However, this raises the question of how to design
a suitable tracing framework and how to provide appropriate language
constructs that can specialize traces to different forms of provenance,
in a type-safe and efficient way.  (In particular,
we cannot simply reuse the provenance polynomials/semirings framework
since it is not able to capture where-provenance~\cite{FTDB2009CheneyCT}.)

Second, and related to the previous point, we need to change not only the
query \emph{behavior} but also the query \emph{result type}.
Specifically, in the type of $q^{\mathsf{where}}$, each field is
replaced with a record consisting of the ordinary data value and its
where-provenance annotation, whereas in $q^{\mathsf{lineage}}$, each
element of a collection in the query result type becomes a pair
consisting of the original data and a \emph{collection} of input row
references.  In previous implementations, we have added this behavior to the typechecker directly (in \Links), or
(in \DSH) used \emph{type families}~\cite{ChaKelPey05} to define the effect of the
where-provenance or lineage transformations at the type level.  In the
case of \DSH, this necessitated subtle changes to the \DSH library, as
well as defining evidence translations at the type and term levels to
convince \Haskell's typechecker that our definitions were type-correct.

Thus, in both \Links and \DSH, our previous work has
shown that it is possible to implement language-integrated provenance,
but the need to manipulate both query expressions and their types
makes this more difficult than we might hope.  Our goal, therefore, is
to identify a small set of language features that addresses all of the
above needs well: we would like to be able to customize the query
behavior to handle multiple forms of provenance, 
while retaining the existing benefits demonstrated by
previous implementations of language-integrated provenance:
specifically type-safety and efficient query generation.

\section{Query traces}\label{sec:query-traces}

In this section we describe what our traces look like through a series
of examples. We show how to rewrite expressions to compute their own
trace in Section~\ref{sec:self-tracing}. As described earlier, the
intent is to compose a trace analysis function with a self-tracing
query and normalize to deforest the trace and only compute the parts
that we actually need.

The \lstinline|trace| keyword causes a query expression to be traced.
For example, \lstinline|trace 2+3| has type \lstinline|Trace Int| and
evaluates to \lstinline|OpPlus{l=Lit 2,r=Lit 3}|. Here,
\lstinline|OpPlus| represents an addition operation and its argument
is a record of the left and right subtraces, and \lstinline|Lit| is
the constructor for traces of literal values. Traces of records are
just records of traces, and traces of lists are just lists of traces,
e.g., tracing the singleton list of the singleton record
\lstinline|[{answer=42}]| results in \lstinline|[{answer=Lit 42}]|.

In general, the trace of an expression with type $A$ has a type where
every base type is replaced by the traced version of the base type,
but all list and record constructors stay the same. We can express
this in \TLinks directly as the type-level function \lstinline|TRACE|
defined in Figure~\ref{fig:tracetf}. We capitalize type-level entities
(except variables) and trace constructors, and write type-level
functions in all uppercase. $\kw{Typerec}$ folds over a type, in this
case the type variable \lstinline|a|. It uses its first three
arguments for base types (in our case replacing \boolty with
\lstinline|Trace Bool|, etc.). The next argument is used if the
argument is a list type and applied to the original element type and
the recursively transformed element type. The next arguments work
similarly for records and trace types.

\begin{figure}
\begin{lstlisting}[style=infigure]
TRACE = λa.Typerec a (Trace Bool, Trace Int, Trace String,
                      λe e'.[e'], λr r'.{r'}, λb t.t)
\end{lstlisting}
  \caption{The type-level function \lstinline!TRACE!.}
  \label{fig:tracetf}
\end{figure}

Tables are typed as lists of records. Their traces reveal that they
are not constants in the query however. Values originating from tables
are marked with the \lstinline|Cell| constructor. For example, the
trace of the agencies table looks like this:
\begin{lstlisting}
[{oid=Cell{tbl="agencies",col="oid",row=1,val=1},
  name=Cell{tbl="agencies",col="name",row=1,val="EdinTours"},
  based_in=Cell{tbl="agencies",col="based_in",row=1,val="Edinburgh"}
  phone=Cell{tbl="agencies",col="phone",row=1,val="412 1200"}},
 {oid=Cell{tbl="agencies",col="oid",row=2,val=2},
  name=Cell{tbl="agencies",col="name",row=2,val="Burns's"},
  based_in=Cell{tbl="agencies",col="based_in",row=2,val="Glasgow"}
  phone=Cell{tbl="agencies",col="phone",row=2,val="607 3000"}}]
\end{lstlisting}

Conditional expressions record the trace of the condition as well as
the trace of the eventually produced result. Polymorphic operations
such as \lstinline|==| record the type they were applied to. The trace
of a \lstinline|for| comprehension carries both the element type of
the input collection and subtraces of both the input and the output.
For example, the following query is a convoluted way to get \lstinline|["Edinburgh"]|.
\begin{lstlisting}
for (a <- table "agencies" ...) where (a.name == "EdinTours") 
  [a.based_in]
\end{lstlisting}
Its trace is shown below. We treat \lstinline|where ($M$) $N$| as syntactic sugar for \lstinline|if $M$ then $N$ else []|.
\begin{lstlisting}
[ If{cond=OpEq String
       {l=For {oid:Int,name:String,...}
            {in={oid=Cell{tbl="agencies",col="oid",...},
                 name=Cell{tbl="agencies",col="name",...},
                 based_in=Cell{tbl="agencies",...}
                 phone=Cell{tbl="agencies",col="phone",...}},
             out=Cell{tbl="agencies",...},
        r=Lit "EdinTours"},
     out=For {oid:Int,name:String,based_in:String,phone:String}
       {in={oid=Cell{tbl="agencies",col="oid",...},
            name=Cell{tbl="agencies",col="name",...},
            based_in=Cell{tbl="agencies",...}
            phone=Cell{tbl="agencies",col="phone",...}},
        out=Cell{tbl="agencies",col="based_in",row=1,
                 val="Edinburgh"}}}} ]
\end{lstlisting}
Note that the variable \lstinline|a| does not appear explicitly in the
trace. Rather, wherever a variable in an expression would produce a
value, we record the subtrace of the value in the trace. Also note
that the trace of this singleton list is still a singleton list, and
the comprehension marker appears on the (singleton) element. This is a
significant deviation from previous work on tracing queries
\cite{cheney14ppdp} which will make trace analysis much easier as
trace analysis functions will not have to deal with variable binding.

\section{Trace analysis}\label{sec:trace-analysis}

Trace analysis functions need to be flexible enough to work with
queries of any type and any shape. The shape of a query, and thus the
depth of its trace, are not even necessarily known until runtime of
the program. Therefore trace analysis functions need to be polymorphic
and recursive. In the following we use $\Lambda$ for term-level type
abstraction, \lstinline|fix| to define recursive values, \kw{typecase}
to branch on types, and \kw{tracecase} to branch on trace
constructors. We will also use generic record operations to work with
records of any number and type of fields. We will describe these in
more detail in Section~\ref{sec:syntax-semantics}.


\subsection{Where-provenance}

\begin{figure}
\begin{lstlisting}[style=infigure]
W = $\lambda$a:Type.{val:a, tbl:String, col:String, row:Int}

WHERE = $\lambda$a:Type.Typerec a (W Bool, W Int, W String,
                  $\lambda$_ b.List b, $\lambda$_ r.Record r, $\lambda$_ b.b)

wherep : $\forall$a.T(TRACE a) -> T(WHERE a)
wherep = fix (wherep:$\forall$a.T(TRACE a) -> T(WHERE a)).$\Lambda$a:Type.
  typecase a of
    List b   => $\lambda$xs.for (x <- xs) [wherep b x]
    Record r => $\lambda$x.rmap$^\texttt{r}$ wherep x
    Trace b  => $\lambda$x.tracecase x of
      Lit y    => fake b y
      If y     => wherep (Trace b) y.out
      For c y  => wherep (Trace b) y.out
      Cell y   => y
      OpPlus y => fake Int (value (Trace Int) x)
      OpEq c y => fake Bool (value (Trace c) x)

fake : $\forall$a.T(a) -> T(W a)
fake = $\Lambda$a.$\lambda$x:T(a).{val=x,tbl="facts",col="alternative",row=-1}
\end{lstlisting}
  \caption{The \texttt{wherep} trace analysis function and supporting definitions.}
  \label{fig:where-ta}
\end{figure}

Where-provenance annotates every cell of a query result with
information about where in the database the value was copied from.
Figure~\ref{fig:where-ta} shows the \lstinline|wherep| trace analysis
function and helpers. On the type level, \lstinline|WHERE| replaces
every base type by a record with fields for the value, table, column,
and row number. For any type $a$, \lstinline|wherep| takes a trace and
returns a where-provenance--annotated value. \lstinline|T()| wraps
type-level computation, as explained later. To recover
where-provenance from a trace, \lstinline|wherep| distinguishes three
cases: did the traced expression have a list type, a record type, or a
base type. In case of a list type, we map \lstinline|wherep| over the
list of subtraces. (We use a comprehension here, but \Links handles
higher-order functions like \lstinline|map| and \lstinline|filter| just fine.) In case of a
record type, we use \kw{rmap} to map \lstinline|wherep| over the
fields of the record of subtraces. In case the original expression was
of some base type $A$, the trace has type \lstinline|Trace $A$|, which
we further analyze using \lstinline|tracecase|. If the trace constructor is \lstinline|Lit|
the value was a constant in the query and we need to mark it with
fake provenance. In the \lstinline|If| and
\lstinline|For| cases, we continue extracting where-provenance from their
output. If the trace constructor is \lstinline|Cell|, the value
originated from the database and already carries the table and column
names and row number. Finally, we associate fake where-provenance with
the results of operators, whose value is computed by the
\lstinline|value| trace analysis function (see
Section~\ref{sec:value}).

\subsection{Value}\label{sec:value}

The \lstinline|value| trace analysis function is the inverse to
tracing. It recovers a plain value from a trace by recomputing values
from operators' subtraces and otherwise throwing away all tracing
information. It is defined in Appendix~\ref{sec:app:value}.

\subsection{Lineage}

This implementation of lineage aims to emulate the behavior of \LLinks, a variant of \Links with built-in support for lineage \citep{FEHRENBACH2018103}.
This is complicated  by the fact that lineage annotations in \LLinks are on rows (or more generally, list elements) but tracing information in \TLinks is on cells.
We need to collect annotations from the trace leaves and pull them up to the nearest enclosing list constructor.

\begin{figure}
\begin{lstlisting}[style=infigure]
L = λa:Type.{data: a, lineage: [{%table%: String, row: Int}}

LINEAGE = λa:Type.Typerec a (Bool, Int, String,
                    λ_ b.List (L b), λ_ r.Record r, λ_ b.b)

lineage : $\forall$a.T(TRACE a) -> T(LINEAGE a)
lineage = fix (lineage:$\forall$a.T(TRACE a) -> T(LINEAGE a)).$\Lambda$a:Type.
  typecase a of
    List b   => λts.for (t <- ts)
                    [{data = lineage b t,
                      lineage = linnotation b t}]
    Record r => λx.rmap$^\texttt{r}$ lineage x
    Trace b  => λx.value (Trace b) x

linnotation : $\forall$a.T(TRACE a) -> [{%table%: String, row: Int}]
linnotation = fix (linnotation: ...).$\Lambda$a:Type.
  typecase a of
    List b   => λts.for (t <- ts) linnotation b t
    Record r => λx.rfold$^{\texttt{Rmap }(\lambda\texttt{\_.[}\langle\texttt{table:String, row:Int}\rangle\texttt{]})\texttt{ r}}$ (++) []
                      (rmap$^\texttt{r}$ linnotation x)
    Trace b  => λt.tracecase t of
      Lit c    => []
      If i     => linnotation (TRACE b) i.out
      For c f  => linnotation (TRACE c) f.in ++
                  linnotation (TRACE b) f.out
      Cell r   => [{%table% = r.%table%, row = r.row}]
      OpEq c e => linnotation (TRACE c) e.left ++
                  linnotation (TRACE c) e.right
      OpPlus p => linnotation (TRACE Int) p.left ++
                  linnotation (TRACE Int) p.right
\end{lstlisting}

  \caption{The \texttt{lineage} trace analysis function and supporting definitions.}
  \label{fig:lineage-ta}
\end{figure}

The \lstinline|LINEAGE| type function changes list types to carry a list of annotations.
On the value level, the implementation is split into two functions: \lstinline|lineage| and \lstinline|linnotation|, as shown in Figure~\ref{fig:lineage-ta}.
The \lstinline|lineage| function matches on the type of its argument and makes (recursive) calls to \lstinline|lineage|, \lstinline|linnotation|, and \lstinline|value| as appropriate to combine annotations and values.
The \lstinline|linnotation| function does the actual work of computing lineage annotations from traces.
The case for lists concatenates the lineage annotations obtained by calling \lstinline|linnotation| on the list elements.
In the case for records, we first use \lstinline|rmap| to map \lstinline|linnotation| over the record, then we use \lstinline|rfold| to flatten the record of lists of lineage annotations into a single list.
Trace constructors have lineage annotations as follows.
Literals do not have lineage.
Conditional expressions have the lineage of their result.
Comprehensions are the interesting case, where we combine lineage annotations from the input with lineage annotations from the output.
Each table cell has the expected initial singleton annotation consisting of its table's name and its row number.
Finally, the operators just collect their arguments' annotations.

There is an issue with this implementation of lineage: we collect duplicate annotations.
Consider the following query:
\begin{lstlisting}
for (x <- table "xs" {a: Int, b: Bool}) [x.a]
\end{lstlisting}
We just project a table to one of its columns.
The lineage of every element of the result should be one of the rows in the table.
If we apply the \lstinline|lineage| trace analysis function to the trace of the above query (at the appropriate type) and normalize, we get this query expression:
\begin{lstlisting}
for (x <- table "xs" {a: Int, b: Bool})
  [{data=x.a, lineage=[{tbl="xs",row=x.oid}] ++
           [{tbl="xs",row=x.oid}] ++ [{tbl="xs",row=x.oid}]}]
\end{lstlisting}
The lineage is correct, but there is too much of it.
Instead of having one annotation with table and row, we have the same annotation three times.
In fact, a similar query on a table with $n$ columns, would produce $n+1$ annotations.
Looking at the trace expression below, we can see the problem.
\begin{lstlisting}
for (x <- table "xs" {a: Int, b: Bool})
  [For {in={a=Cell {tbl="xs", col="a", row=x.oid, val=x.a},
            b=Cell {tbl="xs", col="b", row=x.oid, val=x.b}},
        out=Cell {tbl="xs", col="a", row=x.oid, val=x.a}}]
\end{lstlisting}
The record case combines the annotations from all of the fields, which interacts badly with the tracing of tables, which puts annotations on all of the fields.
There are at least two solutions to this problem that preserve tracing at the level of cells.
The ad-hoc solution is to introduce a set union operator $M \cup N$ with a special normalization rule that reduces to just $M$ if $M$ and $N$ are known to be equal statically.
The proper solution would be to support set and multiset semantics for different portions of the same query and generate \SQL queries that eliminate duplicates where necessary.

\subsection{Normalization and query generation}
To compute the where-provenance of the earlier boat tour agencies
query (let's call it $Q$), we can specialize the \lstinline|wherep| trace analysis function to the traced type of $Q$ and apply it to the traced query itself as follows:
\begin{lstlisting}
wherep (TRACE [{name:String,phone:String}]) (trace $Q$)
\end{lstlisting}
We have seen that traces can get quite big and trace analysis
functions contain features with no obvious counterpart in \SQL. The
rest of this paper shows how exactly tracing works, describes the
language in detail, and discusses normalization to nested relational
calculus, which we can further translate to \SQL. In the end, all of
the trace construction and trace analysis code will be eliminated and
the above code will result in a simple query like the following.
\begin{lstlisting}[language=sql]
SELECT e.name AS name_val, 'externalTours' AS name_tbl,
       'name' AS name_col, e.oid AS name_row, 
       a.phone AS phone_val, 'agencies' AS phone_tbl,
       'phone' AS phone_col, a.oid AS phone_row
  FROM agencies AS a, externaltours AS e
 WHERE a.name = e.name AND e.type = 'boat'
\end{lstlisting}
Note that \Links flattens nested records into top-level columns and
only reassembles records when fetching the results 
\citep{SIGMOD2014CheneyLW}.



\section{\TLinks syntax \& static semantics}\label{sec:syntax-semantics}

\begin{figure*}[tb!]
  \[
    \begin{array}{lrcl}
\text{Contexts} &      \Gamma & \Coloneqq & \cdot \mid \Gamma, \alpha: K \mid \Gamma, x: A \smallskip\\
\text{Kinds} &      K & \Coloneqq & \mathit{Type} \mid \mathit{Row}  \mid K_1 \rightarrow K_2 \smallskip\\
\text{Constructors} &      C, D & \Coloneqq & \texttt{Bool}^* \mid
                                              \inttyc \mid
                                              \texttt{String}^* \mid
                                              \alpha \mid \lambda
                                              \alpha: K.C \mid C\ D
                                              \mid \texttt{List}^*\ C \mid \texttt{Record}^*\ S \mid \texttt{Trace}^*\ C \\
&             & \mid & \kw{Typerec}~ C~ (C_B, C_I, C_S, C_L, C_R, C_T) \smallskip\\
\text{Row Constructors}&      S & \Coloneqq & \cdot \mid l : C; S \mid \rho \mid \kw{Rmap}\ C\ S \smallskip\\
\text{Types} &      A, B & \Coloneqq & \texttt{T}(C) \mid \texttt{Bool} \mid \intty \mid \texttt{String} \mid A \rightarrow B \mid \texttt{List}\ A \mid \texttt{Record}\ R \mid \texttt{Trace}\ A \mid \forall \alpha:K.A \smallskip\\
\text{Rows} &      R   & \Coloneqq & \cdot \mid l : A; R \smallskip\\
\text{Expressions} &      L, M, N & \Coloneqq & c \mid x \mid \lambda
                                                x:A.M \mid M\ N \mid
                                                \Lambda \alpha:K.M
                                                \mid M\ C \mid
                                                \kw{fix}~f:A.M \\
&             & \mid & \kw{if}~L~\kw{then}~M~\kw{else}~N \mid M + N
                       \mid M == N \mid \langle \rangle \mid \langle l = M; N \rangle \mid M.l \\
\text{(Collections)}&             & \mid & [] \mid [M] \mid M \concat N \mid \kw{for}~(x \leftarrow M)~N \mid \kw{table}~n~\langle R \rangle \\
\text{(Traces)}&             & \mid & \texttt{Lit}\ M \mid \texttt{If}\ M \mid \texttt{For}\ C\ M \mid \texttt{Cell}\ M \mid \texttt{OpEq}\ C\ M \mid \texttt{OpPlus}\ M \\
\text{(Trace Analysis)}&             & \mid & \kw{tracecase}~M~\kw{of}~(x.M_L, x.M_I, \alpha.x.M_F, x.M_C, \alpha.x.M_E, x.M_P) \\
\text{(Type Analysis)}&             & \mid & \kw{typecase}~C~\kw{of}\ (M_B, M_I, M_S, \beta.M_L, \rho.M_R, \beta.M_T)  \mid \kw{rmap}^S~L~M \mid \kw{rfold}^S~L~M~N\\
    \end{array}
  \]
  \caption{The syntax of \TLinks.}
  \label{fig:tlinks-syntax}
\end{figure*}
\begin{figure}[tb!]
  \begin{align*}
    (\lambda\alpha:K.C)\ D & ⤳ C[\alpha \coloneqq D] \\
    \kw{Rmap}\ C\ \cdot & \leadsto \cdot \\
    \kw{Rmap}\ C\ (l:D; S) &\leadsto (l: C\ D; \kw{Rmap}\ C\ S)\\
    \kw{Typerec}~\texttt{Bool}~(C_B, \hdots) &\leadsto C_B \\
    \kw{Typerec}~\coll{D}~(\hdots, C_L, \hdots) &\leadsto C_L\ D\ (\kw{Typerec}~D~(\hdots, C_L, \hdots)) \\
\kw{Typerec}~\rec{S}~(\hdots, C_R, \hdots) &\leadsto \\ & \hspace{-4em} C_R\ S\ (\kw{Rmap}\ (\lambda \alpha.\kw{Typerec}~\alpha~(\hdots, C_R, \hdots))\ S)
  \end{align*}

  \caption{Constructor and row constructor computation.}
  \label{fig:constructor-computation}
\end{figure}

\begin{figure*}[tb!]
  \small
  \begin{mathpar}
    \infer
    { Γ ⊢ M : B \\ Γ ⊢ A = B }
    { Γ ⊢ M : A }

    \infer
    { \cdot ⊢ R : \textit{Row} }
    {Γ ⊢ \kw{table}~n~\rec{\textsf{oid}:\intty; R} : \coll{\rec{\textsf{oid}:\intty; R}}}

    \infer
    { Γ ⊢ M : \forall \alpha: \textit{Type}. \texttt{T}(\alpha) \rightarrow \texttt{T}(C\ \alpha) \\
      Γ ⊢ N : \texttt{T}(\texttt{Record}^*\ S)}
    { Γ ⊢ \kw{rmap}^S~M~N : \texttt{T}(\texttt{Record}^*\ (\kw{Rmap}\ C \ S))}

    \infer
    { Γ ⊢ L: \texttt{T}(C) \rightarrow \texttt{T}(C) \rightarrow \texttt{T}(C) \\
      Γ ⊢ M : \texttt{T}(C) \\
      Γ ⊢ N : \texttt{T}(\texttt{Record}^*\ (\kw{Rmap}\ (\lambda\alpha.\alpha \rightarrow C)\ S)) }
    { Γ ⊢ \kw{rfold}^S\ L\ M\ N:\texttt{T}(C) }

    \infer
    { Γ ⊢ C : \textit{Type} \\
      Γ, \alpha: \textit{Type} ⊢ B : \textit{Type} \\
      \beta, \rho, \gamma \notin \textit{Dom}(Γ) \\
      Γ ⊢ M_B : B[\alpha \coloneqq \texttt{Bool}^*] \\
      Γ ⊢ M_I : B[\alpha \coloneqq \inttyc] \\
      Γ ⊢ M_S : B[\alpha \coloneqq \texttt{String}^*] \\
      Γ, \beta : \textit{Type} ⊢ M_L : B[\alpha \coloneqq \texttt{List}^*\ \beta] \\
      Γ, \rho : \textit{Row} ⊢ M_R : B[\alpha \coloneqq \texttt{Record}^*\ \rho] \\
      Γ, \gamma : \textit{Type} ⊢ M_T : B[\alpha \coloneqq \texttt{Trace}^*\ \gamma]
    }
    { Γ ⊢ \kw{typecase}~C~\kw{of}\ (M_B, M_I, M_S, \beta.M_L, \rho.M_R, \gamma.M_T) : B[\alpha \coloneqq C] }
  \end{mathpar}
  \caption{Term formation $Γ ⊢ M : A$.}
  \label{fig:term-formation}
\small
  \begin{mathpar}
    \infer
    { Γ ⊢ c : \intty }
    { Γ ⊢ \texttt{Lit}\ c : \texttt{Trace}\ \intty }

    \infer
    { Γ ⊢ M : \langle \mathsf{cond} : \texttt{Trace}\ \texttt{Bool}, \mathsf{out} : \texttt{Trace}\ A \rangle  }
    { Γ ⊢ \texttt{If}\ M : \texttt{Trace}\ A }

    \infer
    { Γ ⊢ C : \textit{Type} \\ Γ ⊢ M : \langle \mathsf{in} : \texttt{T(TRACE}\ C\texttt{)}, \mathsf{out} : \texttt{Trace}\ A \rangle  }
    { Γ ⊢ \texttt{For}\ C\ M : \texttt{Trace}\ A }

    \infer
    { Γ ⊢ M : \langle \mathsf{tbl} : \texttt{String}, \mathsf{col} : \texttt{String}, \mathsf{row} : \intty, \mathsf{val} : A \rangle  }
    { Γ ⊢ \texttt{Cell}\ M : \texttt{Trace}\ A }

    \infer
    { Γ ⊢ C : \textit{Type} \\ Γ ⊢ M : \langle \mathsf{l} : \texttt{T(TRACE}\ C\texttt{)}, \mathsf{r} : \texttt{T(TRACE}\ C\texttt{)} \rangle}
    { Γ ⊢ \texttt{OpEq}\ C\ M : \texttt{Trace}\ \texttt{Bool} }

    \infer
    { Γ ⊢ M : \texttt{Trace}\ A \\
      Γ, x_L : A ⊢ M_L : B \\
      Γ, x_I : \langle\mathsf{cond}: \texttt{Trace}\ \texttt{Bool}, \mathsf{then}: \texttt{Trace}\ A\rangle ⊢ M_I : B \\
      Γ, \alpha_F: \textit{Type}, x_F : \langle\mathsf{in}: \texttt{T}(\texttt{TRACE}\ \alpha_F), \mathsf{out}: \texttt{Trace}\ A\rangle ⊢ M_F : B \\
      Γ, x_C : \langle\mathsf{tbl}: \texttt{String}, \mathsf{col}: \texttt{String}, \mathsf{row}: \intty, \mathsf{val} : A\rangle ⊢ M_C : B \\
      Γ, \alpha_E : \textit{Type}, x_E : \langle\mathsf{l}: \texttt{T}(\texttt{TRACE}\ \alpha_E), \mathsf{r}: \texttt{T}(\texttt{TRACE}\ \alpha_E)\rangle ⊢ M_E : B \\
      Γ, x_P : \langle\mathsf{l}: \texttt{Trace}\ \intty, \mathsf{r}: \texttt{Trace}\ \intty\rangle ⊢ M_P : B
    }
    { Γ ⊢ \kw{tracecase}~M~\kw{of}\ (x_L.M_L, x_I.M_I, \alpha_F.x_F.M_F, x_C.M_C, \alpha_E.x_E.M_E, x_P.M_P) : B }
  \end{mathpar}

  \caption{\texttt{Trace} introduction and elimination rules (some \texttt{Lit} cases and \texttt{OpPlus} omitted).}
  \label{fig:trace-intro}
\end{figure*}

The syntax of \TLinks is summarized in Figure~\ref{fig:tlinks-syntax}.
\TLinks is a simplification of the core language for \Links queries
introduced by Lindley and Cheney~\cite{TLDI2012LindleyC}. \Links
employs row typing to typecheck record expressions; \emph{row
  variables} can be used to quantify over parts of record types. The
core \Links calculus of~\cite{TLDI2012LindleyC} also covers ordinary
\Links code and the type-and-effect system used to ensure query
expressions only perform operations that are possible on the database.
We omit these aspects as well as more recent extensions such as
algebraic effects and
handlers~\cite{Hillerstrom:2016:LER:2976022.2976033} and session
types~\cite{lfst}.

In addition to the core query language constructs, \TLinks draws
heavily on the \lamiml calculus \citep{Harper:1995:CPU:199448.199475},
which supports \emph{intensional polymorphism}, that is, the
capability to analyze types at run time ($\kw{typecase}$) and define
types by recursion on the structure of other types
($\kw{Typerec}$). Analogous capabilities are also provided for rows,
similar to the \emph{type-level record computation} used in
\Urweb~\cite{PLDI2010Chlipala}.

We use a single context $\Gamma$ for both type variables $\alpha$ and
term variables $x$. In addition to the usual kinds \textit{Type} and
$\rightarrow$, we have \textit{Row}, the kind of rows.
We distinguish type and row constructors from types and rows (again
following \lamiml).  The
difference is that constructors can be subject to type analysis
(e.g. $\kw{typecase}$), and can contain type-level computation
(e.g. $\kw{Typerec}$), but unlike types, cannot employ polymorphism.  Constructors include base type constructors,
type variables (we write $\rho$ for type variables with kind
\textit{Row}), type-level functions and application, list, record, and
trace type constructors, as well as \kw{Typerec} to analyze type
constructors.  Types do not include any computation, but constructors
can be embedded into types using \texttt{T}($C$).  More often than
not, types and constructors are either equivalent or it is obvious
from the context which we are talking about, so we will write, e.g.,
\lstinline|Bool| to mean either the type, or the constructor
\lstinline|Bool|$^*$.  We write $\coll{A}$ and $\coll{C}$ for list
types and constructors and $\rec{R}$ and $\rec{S}$ for record types
and constructors.

Because type constructors can contain nontrivial computation due to
$\kw{Typerec}$, $\kw{Rmap}$ and type-level lambda-abstraction, \TLinks
employs equivalence judgments for types, rows, and their
constructors. 
The more interesting of
the type-level computation rules are shown in
Figure~\ref{fig:constructor-computation}. 
The full set of equivalence rules and type-level computation rules are
relegated to in the appendix due to space limitations.  We conjecture
that type equivalence and typechecking are decidable for \TLinks
(they are for \lamiml) but this remains to be fully investigated.

Most of the typing rules are standard. The more interesting rules
can be found in Figure~\ref{fig:term-formation}. We require that all tables have an
\lstinline|oid| column and otherwise only contain fields of base
types. We can map a sufficiently polymorphic function over a record
using \kw{rmap}. This is reflected on the type level with the row type
constructor \kw{Rmap}. We can fold a homogeneous record into a single
value using \kw{rfold}. Note that we do not specify the order of
folding, so it is best to use a commutative combining function. The
rule for \kw{typecase} is standard, but the improved rule by
\citet{crary_weirich_morrisett_2002} would work as well.

The most representative introduction and elimination rules for the
\texttt{Trace} type can be found in Figure~\ref{fig:trace-intro}. The
constructors for comprehensions and polymorphic operators carry type
information. This type information is brought back in scope when
analyzing traces using \kw{typecase}: the respective branches bind
both a type and a term variable.

\section{The self-tracing transformation}\label{sec:self-tracing}

\begin{figure}[t]
\small
  \begin{align*}
    \trace{x} &= x \\
    \trace{c} &= \texttt{Lit}\ c \\
    \trace{M + N} &= \texttt{OpPlus}\ \langle \texttt{l} = \trace{M}, \texttt{r} = \trace{N}\rangle \\
    \trace{M == (N: \texttt{T}(C))} &= \texttt{OpEq}\ C\ \langle \texttt{l} = \trace{M}, \texttt{r} = \trace{N}\rangle \\
    \trace{\langle \overline{ l = M } \rangle} &= \langle \overline{l = \trace{M}}\rangle \\
    \trace{M.l} &= \trace{M}.l \\
    \trace{\coll{}} &= \coll{} \\
    \trace{\coll{M}} &= \coll{\trace{M}} \\
    \trace{M \concat N} &= \trace{M} \concat \trace{N} \\[0.5\baselineskip]
    \trace{\kw{table}~n~\langle \overline{l : C} \rangle}
              & = \kw{for}\ (y \leftarrow \kw{table}~n~\langle \overline{l : C} \rangle) \hspace{6em} \\
              & \hspace{-6em} \coll{\rec{\overline{l = \texttt{Cell}\rec{\texttt{tbl}=n,\texttt{col}=l,\texttt{row}=y.\mathtt{oid},\texttt{val}=y.l}}}} \\[0.5\baselineskip]
    \ttrace{\kw{for}~(x \leftarrow M:D)~N}{C} &= \kw{for}~(x \leftarrow \trace{M})\\
              & \hspace{-6em} \mathit{dist}(\texttt{TRACE}\ C, \texttt{For}\ D\ \langle \texttt{in} = x, \texttt{out} = \mathbb{H} \rangle,\trace{N}) \\[0.5\baselineskip]
    \ttrace{\kw{if}\ L\ \kw{then}~M~\kw{else}~N}{C}
              &= \kw{if}~\texttt{value}\ (\mathtt{Trace}\ \mathtt{Bool})\ \trace{L}\\
              & \hspace{-6em} \kw{then}~ \mathit{dist}(\texttt{TRACE}\ C, \texttt{If}\langle \texttt{cond} = \trace{L}, \texttt{out} = \mathbb{H} \rangle, \trace{M}) \\
              & \hspace{-6em} \kw{else}~ \mathit{dist}(\texttt{TRACE}\ C, \texttt{If}\langle \texttt{cond} = \trace{L}, \texttt{out} = \mathbb{H} \rangle, \trace{N}) \\
  \end{align*}
  \begin{align*}
    \mathit{dist}(\langle \overline{l:C} \rangle, k, r) &=\langle \overline{ l = \mathit{dist}(C,k,r.l) }\rangle \\
    \mathit{dist}(\coll{C}, k, l) &= \kw{for}\ (x \leftarrow l)\ \coll{\mathit{dist}(C, k, x)}\\
    \mathit{dist}(\texttt{Trace}\ C, k, t) &= k[\mathbb{H} \coloneqq t]
  \end{align*}
  \caption{The self-tracing transformation.}
  \label{fig:self-tracing}
\end{figure}

The self-tracing transformation turns a \emph{normalized} query
expression into an expression that produces a trace of its own
execution. As seen in Figure~\ref{fig:self-tracing}, most cases are
straightforward. Variables inside a self-tracing query refer to their
subtrace directly. Tables are the only source of \texttt{Cell} trace
constructors. Comprehensions and conditionals need to distribute a
trace constructor over a subtrace of any shape including lists and
record types. We accomplish this with the meta-level helper function
$\mathit{dist}$. It takes a type, an expression with a hole
$\mathbb{H}$ in it, and a value of the given type and traverses lists
and records until it reaches the leaves and wraps the expression
with the hole around them. Alternatively, we could have written
\textit{dist} as a \Links function with the type
\begin{lstlisting}
dist: $\forall$a. ($\forall$b. Trace b -> Trace b) -> TRACE a -> TRACE a
\end{lstlisting}
but using it requires a lot of
boilerplate code for handling impossible cases, so we
prefer the definition in Figure~\ref{fig:self-tracing}.

With these definitions in hand, we check that the self-tracing
transformation preserves well-formedness.  Note that the type-level
function \texttt{TRACE} is needed to state these properties.  Proof
details are in the appendix.
\begin{lemma}\label{lem:dist-type-correctness}
  For all types $C$ that can appear in query types (base types, list types, closed record types), all expressions $k$ with a hole $\mathbb{H}$ that have type $\texttt{Trace}\ D$ assuming the hole $\mathbb{H}$ has type $\texttt{Trace}\ D$,
  and all expressions $M$ of type $\mathtt{TRACE}\ C$,
$\mathit{dist}(\texttt{TRACE}\ C, k, M)$ has type $\texttt{TRACE}\ C$.
\end{lemma}

\begin{theorem}\label{thm:trace-type-correctness-2}
  If $\Gamma \vdash M : A$ then for all $C$,
  if $\Gamma \vdash A = \mathtt{T}(C)$
  then $\trace{\Gamma} \vdash \trace{M} : \mathtt{T}(\mathtt{TRACE}\ C)$,
  where $\Gamma$ is a context that maps all term variables to closed records with fields of base type
  and $M$ is a plain \Links query term in normal form.
\end{theorem}

\section{Normalization}\label{sec:normalization}

Our ultimate goal is to translate \TLinks queries --- including
provenance extraction by trace analysis --- to \SQL. We know from
previous work \citep{JCSS1996Wong,DBPL2009Cooper,TLDI2012LindleyC,SIGMOD2014CheneyLW} that \NRC
expressions, extended with sum types and higher-order functions, can be translated
to \SQL as long as their return type is nested relational. In this
section, we extend query normalization to deal with the new features
for tracing and trace analysis.

We show progress and preservation which imply the existence of a
partial normalization function. Unlike standard progress and
preservation, we do not normalize to values, but to a normal form that
includes table references and residual query code which is ultimately
translated to \SQL queries.

We cannot show strong normalization, since we require recursive
functions to be able to analyze arbitrary queries.

\subsection{Reduction rules}\label{sec:reduction-rules}

\begin{figure}
\small
  \begin{align*}
    \kw{fix}~f.M &\leadsto M[f \coloneqq \kw{fix}~f.M] \\
    (\Lambda \alpha.M)\ C &\leadsto M[\alpha \coloneqq C] \\
    \kw{rmap}^{(\overline{l_i:C_i})}\ M\ N &\leadsto \langle \overline{l_i = (M\ C_i)\ N.l_i } \rangle \\
    \kw{rfold}^{(\overline{l_i: C_i})}\ L\ M\ N & \leadsto L\ N.l_1\ (L\ N.l_2 \hdots (L\ N.l_n\ M) \hdots)
  \end{align*}\vspace*{-1.5\baselineskip}%
  \begin{align*}
    \kw{tracecase}~\texttt{Lit}~M~\kw{of}\ (x.M_L,\hdots) &\leadsto M_L[x \coloneqq M] \\
    \kw{tracecase}~\texttt{For}~C~M~\kw{of}\ (\hdots, \alpha x.M_F, \hdots) & \leadsto M_F[\alpha \coloneqq C, x \coloneqq M] \\
    \kw{typecase}~\texttt{Bool}~\kw{of}\ (M_B,\hdots)
                     &\leadsto M_B  \\
    \kw{typecase}~\coll{C}~\kw{of}\ (\hdots, \beta.M_L, \hdots)
                     &\leadsto M_L[\beta \coloneqq C] \\
    \kw{typecase}~\rec{S}~\kw{of}\ (\hdots, \rho.M_R, \hdots)
                     &\leadsto M_R[\rho \coloneqq S]
  \end{align*}

  \caption{Normalization $\beta$-rules.}
  \label{fig:normalization-beta}
\end{figure}

\TLinks uses the same general approach to normalization as plain
\Links \citep{SIGMOD2014CheneyLW}. We define a relation $\leadsto$
between terms. Most rules are standard.
Figure~\ref{fig:normalization-beta} shows the $\beta$-rules for the
new \TLinks features. Since constructors can appear in terms, e.g.,
typecase, we also need to normalize constructors. We use the same
rules as for type-level computation
(Figure~\ref{fig:constructor-computation}). We also need to add
commuting conversions to, e.g., lift if-then-else out of tracecase, to
expose additional $\beta$ reductions. The full rules can
be found in the appendix in Figures
\ref{fig:app:constructor-computation},
\ref{fig:app:normalization-beta}, \ref{fig:app:normalization-cc}, 
\ref{fig:app:normalization-rules-congruence}.

Unlike plain \Links, we allow recursion in queries and unroll
fixpoints as necessary. It is up to the programmer to ensure that
their functions terminate.
Record map and record fold inspect their row
constructor argument only. Record map evaluates to a new record where we
apply the given function to each field's type and value. Record fold
applies the given function to the accumulator and every record field's
value successively. We evaluate tracecase and typecase by
reducing to the appropriate branch and substituting terms and
constructors for term and type variables.


\subsection{Preservation}\label{sec:preservation}

To prove preservation we will need several substitution lemmas.
Substitution of variables in terms, type variables in types, and type
variables in terms are standard for \lamiml~\citep{morrisett1995compiling,crary_weirich_morrisett_2002}. We
additionally need variants for row constructors: substitution of row
variables in types and substitution of row variables in terms. We also
need standard context manipulation lemmas for weakening and swapping
the order of unrelated variables. For details, see
Appendix~\ref{sec:app:additional}.

Now we can prove that the reduction relation $\leadsto$ preserves the kinds of constructors and the types of terms.

\begin{lemma}\label{lem:preservation-constructor}
  For all type constructors $C$ and row constructors $S$, contexts $\Gamma$, and kinds $K$, if $\Gamma \vdash C : K$ and $C \leadsto C'$, then $\Gamma \vdash C' : K$ and if $\Gamma \vdash S : K$ and $S \leadsto S'$, then $\Gamma \vdash S' : K$.
\end{lemma}
The proof is straightforward by induction on the kinding derivation. For details, see Section~\ref{sec:prf:preservation-constructor}.

\begin{lemma}[Preservation]\label{lem:preservation}
  For all terms $M$ and $M'$, contexts $\Gamma$, and types $A$, if $\Gamma \vdash M : A$ and $M \leadsto M'$, then $\Gamma \vdash M' : A$.
\end{lemma}
The proof is by induction on the typing derivation
$\Gamma \vdash M : A$. The cases for record map and record fold
require type equivalence under type-level computation. The cases for
typecase require the more exotic substitution lemmas from before. See
Section~\ref{sec:prf:preservation} for the proof.

\subsection{Normal form}\label{sec:normal-form}

The goal of normalization is to perform partial evaluation of those parts of the program that are independent of database values.
In particular, we look to eliminate all language constructs which we cannot translate to \SQL.
The \TLinks normal form (Figure~\ref{fig:tlinks-normal-form}) describes what terms look like after exhaustive application of the rewriting rules.
It appears we were not successful, seeing that record map and fold, tracecase, and typecase are all still present.
However, the normal form grammar splits constructors into normal constructors $C$ and neutral constructors $E$, and row constructors into normal row constructors $S$ and neutral row constructors $U$.

\begin{figure*}
  \[
    \begin{array}{lrcl}
      \text{Normal constructors} & C & \Coloneqq & E \mid \texttt{Bool}^* \mid \inttyc \mid \texttt{String}^* \mid \lambda \alpha: K.C \mid \texttt{List}^*\ C \mid \texttt{Record}^*\ S \mid \texttt{Trace}^*\ C \\
      \text{Neutral constructors} & E & \Coloneqq & \alpha \mid E\ C \mid \kw{Typerec}~E~ (C_B, C_I, C_S, C_L, C_R, C_T) \\
      \text{Normal row constructors} & S & \Coloneqq & U \mid \cdot \mid l : C; S \\
      \text{Neutral row constructors} & U & \Coloneqq & \rho \mid l : C; U \mid \kw{Rmap}\ C\ U \\

      \text{Normal terms} & M, N & \Coloneqq & F \mid c \mid \lambda x:A.M \mid \Lambda
                         \alpha:K.M  \mid
                         \kw{if}~H~\kw{then}~M~\kw{else}~N \mid M + N \mid\rec{} \mid \rec{l = M; N} \\
                            & & \mid & \coll{} \mid \coll{M} \mid M \concat N \mid \kw{for}~(x \leftarrow T)~N \mid \kw{table}\ n\ \rec{R} \\
                            & & \mid & \texttt{Lit}\ M \mid \texttt{If}\ M \mid \texttt{For}\ C\ M \mid \texttt{Cell}\ M \mid \texttt{OpEq}\ C\ M \mid \texttt{OpPlus}\ M \\
      \text{Neutral terms} & F & \Coloneqq & x \mid P.l \mid F\ M \mid F\ C \mid \kw{rfold}^U~L~M~N \mid \kw{rmap}^U~M~N \\
                            & & \mid & \kw{tracecase}~F~\kw{of}~(x.M_L, x.M_I, \alpha.x.M_F, x.M_C, \alpha.x.M_E, x.M_P) \\
                            & & \mid & \kw{typecase}~E~\kw{of}~(M_B, M_I, M_S, \beta.M_L, \rho.M_R, \beta.M_T) \\
      \text{Neutral conditional} & H & \Coloneqq & F \mid M == N \\
      \text{Neutral projection} & P & \Coloneqq & F \mid \kw{rmap}^U~M~N \\
      \text{Neutral table} & T & \Coloneqq & F \mid \kw{table}~n~\langle R \rangle
    \end{array}
  \]
  \caption{\TLinks normal form.}
  \label{fig:tlinks-normal-form}
\end{figure*}

\begin{remark}\label{remark:constructors-free-type-variable}
  Neutral constructors $E$ and neutral row constructors $U$ always contain at least one free type variable $\alpha$ or $\rho$ and those are the only base cases for their respective sort.
\end{remark}
We will later use the above to show that some term forms are impossible within queries.
Queries do not contain free type variables, so $E$ and $U$ collapse into nothing, and terms built from $E$ and $U$ (like \lstinline|rmap|) cannot appear.

Similarly, terms are split into normal terms $M$ and neutral terms $F$.
The latter are stuck on a free variable $x$, a stuck constructor $E$, or a stuck row constructor $U$.
We will later argue that inside a query all variables are references to tables and therefore restricted to be base types or records with fields of base types.
This means they cannot be functions or trace constructors and therefore record map, record fold, tracecase, and typecase do not actually appear in normal form queries.


\subsection{Progress}\label{sec:progress}

Progress states that well typed terms either already are in the normal
form described in the previous section or that there is a further
reduction step possible. Reduction preserves typing, so we can keep
reducing until we reach normal form and thus obtain a partial
normalization function.

Like preservation, progress is split into two lemmas: one for constructors and row constructors and one for terms.

\begin{lemma}\label{lem:progress-constructor}
  All well-kinded type constructors $C$ and row constructors $S$, are either in normal form, or there is a type constructor $C'$ with $C \leadsto C'$, or row constructor $S'$ with $S \leadsto S'$.
\end{lemma}
The proof is straightforward by induction on the kinding derivations of $C$ and $S$ (see Section~\ref{sec:prf:progress-constructor}).

\begin{lemma}[Progress]\label{lem:progress}
  For all well-typed terms $M$, either $M$ is in normal form, or there is a term $M'$ with $M \leadsto M'$.
\end{lemma}
The proof (see Section~\ref{sec:prf:progress}) is by induction on the
typing derivation of $M$. Most nontrivial cases have three parts:
reduce in subterms via congruence rules; a
$\beta$-rule applies; or a commuting conversion applies.


\subsection{Normal terms with query types are NRC}

\TLinks normal form still includes language constructs such as
typecase, which do not have an obvious \SQL counterpart. In this
section, we will argue that these cannot actually occur in a query.
Queries are closed expressions with nested relational type. Inside a
query, all variables refer to tables. This is captured in the
following definition of query contexts.

\begin{definition}[Query context]\label{def:query-context}
~  \begin{itemize}
  \item The empty context $\cdot$ is a query context.
  \item The context $\Gamma, x:\langle \overline{l_i : A_i} \rangle$ is a query context, if $\Gamma$ is a query context, $x$ is not bound in $\Gamma$ already, and each type $A_i$ is a base type.
  \end{itemize}
\end{definition}

The \TLinks normal form includes neutral terms $F$, which include
record map and fold, tracecase, and typecase. With the following
Lemma, we will further restrict which terms $F$ can appear in queries
to just variables $x$ and projections $x.l$.

\begin{lemma}\label{lem:f-collapses}
  A term in neutral form $F$ that is well-typed in a query context $\Gamma$, is of the form $x$ or $x.l$.
\end{lemma}
\begin{proof}
  By induction on the typing derivation. The term cannot be a record
  fold or typecase, because those necessarily contain a (row) type
  variable (Remark~\ref{remark:constructors-free-type-variable}),
  which is unbound in the query context $\Gamma$
  (Definition~\ref{def:query-context}). It cannot be a term
  application, type application, or tracecase, because the term in
  function position or the scrutinee, by IH, is of the form $x$ or
  $x.l$, both of which are ill-typed given that the query context
  $\Gamma$ does not contain function types, polymorphic types, or
  trace types. Projections $P.l$ are of the form $F.l$ or
  $(\kw{rmap}^U\ M\ N).l$. The former case reduces by IH to $x.l$ or
  $x.l'.l$, the first of which is okay, and the second is ill-typed.
  The latter case is impossible, because $U$ necessarily contains a
  row variable and would therefore be ill-typed. This leaves variables
  $x$ and projections of variables $x.l$.
\end{proof}

\begin{figure}
  \footnotesize
  \[
    \begin{array}{lrcl}
      \text{Types} & A & \Coloneqq & \mathtt{Bool} \mid \mathtt{Int} \mid \mathtt{String} \mid [A] \mid \langle\overline{l:A}\rangle \\
      \text{Terms} & M, N, L & \Coloneqq & c \mid x \mid \langle
                                           \overline{l = M} \rangle
                                           \mid M.l \mid M + N \mid M
                                           == N \\
&&\mid& \kw{if}~L~\kw{then}~M~\kw{else}~N \mid \kw{table}\ n\ \langle \overline{l: A} \rangle\\
                   & & \mid & [] \mid [M] \mid M \concat N \mid \kw{for}~(x \leftarrow N)~M 
    \end{array}
  \]
  \caption{Target normal form for queries: \NRC.}
  \label{fig:normalization-target-normal-form}
\end{figure}

Finally, we can use this to show that query terms in \TLinks normal
form are actually in nested relational calculus already.

\begin{theorem}\label{thm:nf-query-nrc}
  If $M$ is a term in normal form with a nested relational type in a query context $\Gamma$, then $M$ is in the nested relational calculus (Figure
\ref{fig:normalization-target-normal-form}).
\end{theorem}
The proof (Section~\ref{sec:prf:nf-query-nrc}) is by induction on the
typing derivation, making use of query contexts
(Definition~\ref{def:query-context}),
Remark~\ref{remark:constructors-free-type-variable}, and
Lemma~\ref{lem:f-collapses}.

From here, we can use previous work such as query shredding
\citep{SIGMOD2014CheneyLW} or flattening \citep{SIGMOD2015UlrichG} to
produce \SQL.


\section{Related work}

Extracting provenance from traces is not a new idea
\citep{acar13jcs,cheney14ppdp,misinterpretingsql}.
What makes our work different is that traces and trace analysis are
defined in the language itself. In combination with query
normalization, this makes \TLinks the first, to our knowledge, system
that can execute user-defined query trace analysis on the database.

The traces in \TLinks take inspiration from work on \emph{slicing} of
database queries and programs
\citep{cheney14ppdp,ICFP2012PereraACL,Ricciotti:2017:IFP:3136534.3110258}.
Compared to theirs, our traces contain less information. Some
information would be easy to add, like concatenation operations or
projections. Other information requires changing the structure of
traces in a more invasive way. In particular, our traces are
cell-level only and do not include information about the binding
structure of queries. We also trace only after a first normalization
phase, so traces do not include information about, e.g.,
functions in the original query code. Expression-shaped
traces with explicit representation of variables like those proposed
by \citet{cheney14ppdp}, seem to make writing well-typed
analysis functions more difficult.

\citet{misinterpretingsql} trace query execution and show how
non-standard interpretations of the \SQL semantics produce
where-provenance and lineage instead of query results. They decompose
traces into a static part that resembles the shape of the query, and a
dynamic part which records control-flow decisions made by the database
during query execution. Their work extends to \SQL features like
grouping and aggregation that are not implemented in \Links, let alone
traced in \TLinks. Unlike in \TLinks, alternative interpretation of
queries happens after a trace has been recorded. Thus it is not
possible for the database to optimize, for example, filters based on
provenance information.

\TLinks builds on \lamiml \citep{morrisett1995compiling}.  The \lamr
calculus of Crary et
al.~\citep{crary_weirich_morrisett_2002} improves on \lamiml in making runtime type information explicit, avoiding passing types where unnecessary, and improving the ergonomics of the typecase typing rule by refining types in context.
An actual implementation would benefit from these improvements.

\TLinks features generic record programming in the form of record
mapping and folding. \Urweb \citep{POPL2015Chlipala} features ``first
class, type-level names and records'' \citep{PLDI2010Chlipala}. Its
generic and metaprogramming features seem suitable for our needs, but \Urweb currently lacks the advanced query normalization features we require. Type inference for \TLinks is an open problem. Type inference
for \Urweb is undecidable. However, \citet{PLDI2010Chlipala} claims
that heuristics work well-enough in practice to mostly avoid proof
terms and complex type annotations. Maybe this could be a model for
\TLinks, too.

While we present this work as an extension of \Links and its query
normalization rules, it is conceivable that one could similarly extend
other systems such as the flattening transformation implemented in
\DSH \citep{SIGMOD2015UlrichG}, or the tagless final implementation of
query shredding by \citet{suzuki16pepm}.

\section{Conclusions}

Language-integrated support for queries and their provenance seems
promising, but currently requires nontrivial interventions in the
language implementation or sophisticated metaprogramming
capabilities.  In this paper, we take a step towards making
language-integrated provenance easily customizable by factoring
provenance translations into a self-tracing transformation (that can
be implemented once and for all) and generic programming and trace
analysis capabilities (that can be used to implement different
provenance transformations).  Nevertheless, our work so far is a
foundational language design and more remains to be done to make it
practical.  We have not said anything about typechecking or inference
or, more generally, how \TLinks interfaces with the rest of \Links.  The expressiveness
and generality of our approach to traces needs to be tested further,
by using it to implement other forms of provenance.  Conversely, the
features of \TLinks may have further applications beyond provenance,
like the row-generic programming techniques employed by \Urweb. In
particular, even without traces and trace analysis, our results
extend the theory of conservativity for \NRC queries to 
normalization of typecase and typerec constructs (albeit in the
presence of nonterminating fixedpoint computations).  Sharpening these
results to ensure termination of trace analysis functions would also
be an interesting challenge.

\paragraph*{Acknowledgments}
This work was supported by a Google Faculty Research Award and ERC Consolidator Grant Skye (grant number
  \grantnum{ERC}{682315}).

\balance
\bibliography{paper}

\appendix
\onecolumn

\section{The value trace analysis function}\label{sec:app:value}
\begin{lstlisting}[style=infigure]
VALUE = $\lambda$a:Type.Typerec a (Bool, Int, String, $\lambda$_ b.List b, $\lambda$_ r.Record r, $\lambda$c _.c)

value : $\forall$a.T(a) -> T(VALUE a)
value = fix (value: $\forall$a.T(a) -> T(VALUE a)).$\Lambda$a:Type.
  typecase a of
    Bool     => $\lambda$x:Bool.x
    Int      => $\lambda$x:Int.x
    String   => $\lambda$x:String.x
    List b   => $\lambda$x:List b.for (y <- x) [value b y]
    Record r => $\lambda$x:Record r.rmap$^\texttt{r}$ value x
    Trace b  => $\lambda$x:Trace b.tracecase x of
      Lit y    => y
      If y     => value (Trace b) y.out
      For c y  => value (Trace b) y.out
      Cell y   => y.data
      OpPlus y => value (Trace Int) y.left + value (Trace Int) y.right
      OpEq c y => value (TRACE c) y.left == value (TRACE c) y.right
\end{lstlisting}

\section{Full formalization of \TLinks}

\subsection{Kinding judgments}
  \begin{itemize}
  \item Figure~\ref{fig:app:well-formed-contexts} gives the rules for
    well-typed contexts ($Γ, \alpha : K \text{ is well-formed}$)
\item Figure~\ref{fig:app:constructor-formation} defines the
  well-formedness judgment for type constructors ($Γ ⊢ C: K$)
\item Figure~\ref{fig:app:type-formation} defines the well-formedness
  judgment for types ($Γ ⊢ A: K$)
  \end{itemize}

\begin{figure*}
  \footnotesize
  \begin{mathpar}
    \infer
    { }
    { \cdot \text{ is well-formed} }

    \infer
    {Γ ⊢ A : \textit{Type} \and x \notin \textit{Dom}(Γ)}
    {Γ, x : A \text{ is well-formed}}

    \infer
    {Γ \text{ is well-formed} \and \alpha \notin \textit{Dom}(\Gamma) }
    {Γ, \alpha : K \text{ is well-formed}}
  \end{mathpar}
  \caption{Well-formed contexts $\Gamma$.}
  \label{fig:app:well-formed-contexts}
\end{figure*}

\begin{figure*}
  \footnotesize
  \begin{mathpar}
    \infer
    { Γ \text{ well-formed} }
    { Γ ⊢ \texttt{Bool}^* : \textit{Type} }

    \infer
    { Γ \text{ well-formed} }
    { Γ ⊢ \inttyc : \textit{Type} }

    \infer
    { Γ \text{ well-formed} }
    { Γ ⊢ \texttt{String}^* : \textit{Type} }

    \infer
    { Γ(\alpha) = K }
    { Γ ⊢ \alpha: K }

    \infer
    { Γ, \alpha: K_1 ⊢ C: K_2 }
    { Γ ⊢ \lambda \alpha:K_1.C : K_1 \rightarrow K_2 }

    \infer
    { Γ ⊢ C : K_1 \rightarrow K_2 \\ Γ ⊢ D : K_1 }
    { Γ ⊢ C\ D: K_2 }

    \infer
    { Γ ⊢ C : \textit{Type} }
    { Γ ⊢ \texttt{List}^*\ C : \textit{Type} }

    \infer
    { Γ ⊢ S : \textit{Row} }
    { Γ ⊢ \texttt{Record}^*\ S: \textit{Type} }

    \infer
    { Γ ⊢ C : \textit{Type} }
    { Γ ⊢ \texttt{Trace}^*\ C : \textit{Type} }

    \infer
    { Γ ⊢ C : \textit{Type} \\
      Γ ⊢ C_B : K \\
      Γ ⊢ C_I : K \\
      Γ ⊢ C_S : K \\
      Γ ⊢ C_L : \textit{Type} \rightarrow K \rightarrow K \\
      Γ ⊢ C_R : \textit{Row} \rightarrow \textit{Row} \rightarrow K \\
      Γ ⊢ C_T : \textit{Type} \rightarrow K \rightarrow K }
    { Γ ⊢ \kw{Typerec}~C~(C_B, C_I, C_S, C_L, C_R, C_T) : K }

    \infer
    { Γ \text{ well-formed} }
    { Γ ⊢ \cdot : \textit{Row} }

    \infer
    { Γ ⊢ C:\textit{Type} \\ Γ ⊢ S : \textit{Row} }
    { Γ ⊢ l:C;S : \textit{Row} }

    \infer
    { Γ ⊢ C : \textit{Type} \rightarrow \textit{Type} \\ Γ ⊢ S : \textit{Row} }
    { Γ ⊢ \kw{Rmap}\ C\ S : \textit{Row} }
  \end{mathpar}

  \caption{Constructor and row constructor kinding.}
  \label{fig:app:constructor-formation}
\end{figure*}

\begin{figure*}
  \footnotesize
  \begin{mathpar}
    \infer
    { Γ ⊢ C: \textit{Type} }
    { Γ ⊢ \texttt{T}(C): \textit{Type} }

    \infer
    { Γ \text{ well-formed} }
    { Γ ⊢ \texttt{Bool}: \textit{Type} }

    \infer
    { Γ \text{ well-formed} }
    { Γ ⊢ \intty: \textit{Type} }

    \infer
    { Γ \text{ well-formed} }
    { Γ ⊢ \texttt{String}: \textit{Type} }


    \infer
    { Γ, \alpha:K ⊢ A: \textit{Type} \\ \alpha \notin \textit{Dom}(Γ) }
    { Γ ⊢ \forall \alpha:K.A : \textit{Type} }

    \infer
    { Γ ⊢ A : \textit{Type} \\ Γ ⊢ B : \textit{Type} }
    { Γ ⊢ A \rightarrow B : \textit{Type} }

    \infer
    { Γ ⊢ A : \textit{Type} }
    { Γ ⊢ \texttt{List}\ A : \textit{Type} }

    \infer
    { Γ ⊢ R : \textit{Row} }
    { Γ ⊢ \texttt{Record}\ R : \textit{Type} }

    \infer
    { Γ ⊢ A : \textit{Type} }
    { Γ ⊢ \texttt{Trace}\ A : \textit{Type} }

    \infer
    { Γ ⊢ S: \textit{Row} }
    { Γ ⊢ \texttt{T}(S) : \textit{Row}}

    \infer
    { Γ \text{ well-formed} }
    { Γ ⊢ \cdot : \textit{Row} }

    \infer
    { Γ ⊢ A : \textit{Type} \\ Γ ⊢ R : \textit{Row} }
    { Γ ⊢ l:A;R : \textit{Row}}
  \end{mathpar}
  \caption{Type and row type kinding.}
  \label{fig:app:type-formation}
\end{figure*}

\subsection{Type-level computation and equivalence}

\begin{itemize}
\item Figure~\ref{fig:app:constructor-computation} defines the
  reduction relation for type and row constructors ($C \leadsto C'$,
  $S \leadsto S'$)
\item Figure~\ref{fig:app:constructor-equivalence} defines equivalence
  for type and row constructors ($Γ ⊢ C = C' : K$, $Γ ⊢ S = S' : K$)
\item Figure~\ref{fig:app:type-equivalence} defines type and row
  equivalence ($Γ ⊢ A = B: K$, $Γ ⊢ S = S': \mathit{Type}$)
\end{itemize}

\begin{figure*}
  \footnotesize
  \begin{align*}
    S \leadsto S' \Rightarrow l:C; S &\leadsto l:C; S' \\
    C \leadsto C' \Rightarrow l:C; S &\leadsto l:C'; S \\
    \\
    C \leadsto C' \Rightarrow C\ D &\leadsto C'\ D \\
    D \leadsto D' \Rightarrow C\ D &\leadsto C\ D' \\
    (\lambda\alpha:K.C)\ D & \leadsto C[\alpha \coloneqq D] \\
    \\
    C \leadsto C' \Rightarrow \lambda \alpha:K.C &\leadsto \lambda \alpha:K.C' \\
    C \leadsto C' \Rightarrow \mathtt{List}^*\ C &\leadsto \mathtt{List}^*\ C' \\
    C \leadsto C' \Rightarrow \mathtt{Trace}^*\ C &\leadsto \mathtt{Trace}^*\ C' \\
    S \leadsto S' \Rightarrow \mathtt{Record}^*\ S &\leadsto \mathtt{Record}^*\ S' \\
    \\
    \kw{Rmap}\ C\ \cdot & \leadsto \cdot \\
    \kw{Rmap}\ C\ (l:D; S) &\leadsto (l: C\ D; \kw{Rmap}\ C\ S)\\
    S \leadsto S' \Rightarrow \kw{Rmap}\ C\ S &\leadsto \kw{Rmap}\ C\ S' \\
    C \leadsto C' \Rightarrow \kw{Rmap}\ C\ S &\leadsto \kw{Rmap}\ C'\ S \\
    \\
    C \leadsto C' \Rightarrow \kw{Typerec}~C~(C_B, C_I, C_S, C_L, C_R, C_T) &\leadsto \kw{Typerec}~C'~(C_B, C_I, C_S, C_L, C_R, C_T) \\
    C_B \leadsto C_B' \Rightarrow \kw{Typerec}~C~(C_B, C_I, C_S, C_L, C_R, C_T) &\leadsto \kw{Typerec}~C~(C_B', C_I, C_S, C_L, C_R, C_T) \\
    & \vdots \\
    \kw{Typerec}~\texttt{Bool}^*~(C_B, C_I, C_S, C_L, C_R, C_T) &\leadsto C_B \\
    \kw{Typerec}~\inttyc~(C_B, C_I, C_S, C_L, C_R, C_T) &\leadsto C_I \\
    \kw{Typerec}~\texttt{String}^*~(C_B, C_I, C_S, C_L, C_R, C_T) &\leadsto C_S \\
    \kw{Typerec}~\texttt{List}^*~D~(C_B, C_I, C_S, C_L, C_R, C_T) &\leadsto C_L\ D\ (\kw{Typerec}~D~(C_B, C_I, C_S, C_L, C_R, C_T)) \\
    \kw{Typerec}~\texttt{Record}^*~S~(C_B, C_I, C_S, C_L, C_R, C_T) &\leadsto C_R\ S\ (\kw{Rmap}\ (\lambda \alpha.\kw{Typerec}~\alpha~(C_B, C_I, C_S, C_L, C_R, C_T))\ S) \\
    \kw{Typerec}~\texttt{Trace}^*~D~(C_B, C_I, C_S, C_L, C_R, C_T) &\leadsto C_T\ D\ (\kw{Typerec}~D~(C_B, C_I, C_S, C_L, C_R, C_T))
  \end{align*}

  \caption{Constructor and row constructor computation.}
  \label{fig:app:constructor-computation}
\end{figure*}

\begin{figure*}
  \footnotesize
  \begin{mathpar}
    \infer
    { Γ ⊢ C : K }
    { Γ ⊢ C = C : K }

    \infer
    { Γ ⊢ D = C : K }
    { Γ ⊢ C = D : K }

    \infer
    { Γ ⊢ C = C' : K \\ Γ ⊢ C' = C'' : K }
    { Γ ⊢ C = C'' : K }

    \infer 
    { Γ ⊢ C : K \rightarrow K' }
    { Γ ⊢ \lambda\alpha:K.C\ \alpha = C: K \rightarrow K' }

    \infer
    { Γ,\alpha:K ⊢ C = D: K' \\ \alpha \notin \textit{Dom}(Γ) }
    { Γ ⊢ \lambda\alpha:K.C = \lambda\alpha:K.D: K \rightarrow K' }

    \infer
    { Γ ⊢ C = C' : K' \rightarrow K \\ Γ ⊢ D = D' : K' }
    { Γ ⊢ C\ D = C'\ D' : K }

    \infer
    { Γ ⊢ C = D : K }
    { Γ ⊢ \texttt{List}^*\ C = \texttt{List}^*\ D : K }

    \infer
    { Γ ⊢ S = S' : K }
    { Γ ⊢ \texttt{Record}^*\ S = \texttt{Record}^*\ S' : K }

    \infer
    { Γ ⊢ C = D : K }
    { Γ ⊢ \texttt{Trace}^*\ C = \texttt{Trace}^*\ D : K }

    \infer
    { Γ ⊢ C : K \\ Γ ⊢ D : K \\ C \leadsto D }
    { Γ ⊢ C = D : K }

    \infer
    { Γ \text{ well-formed} }
    { Γ ⊢ \cdot = \cdot : \textit{Row} }

    \infer
    { Γ ⊢ C = D: \textit{Type} \\ Γ ⊢ S = S': \textit{Row} }
    { Γ ⊢ (l:C; S) = (l:D; S') : \textit{Row} }

    \infer
    { Γ ⊢ C = D : \textit{Type} \rightarrow \textit{Type} \\ Γ ⊢ S = S' : \textit{Row} }
    { Γ ⊢ \kw{Rmap}\ C\ S = \kw{Rmap}\ D\ S' : \textit{Row} }

    \infer
    { Γ ⊢ C = C' : K \\
      Γ ⊢ C_B = C_B' : K \\
      Γ ⊢ C_I = C_I' : K \\
      Γ ⊢ C_S = C_S' : K \\
      Γ ⊢ C_L = C_L' : \textit{Type} \rightarrow K \rightarrow K \\
      Γ ⊢ C_R = C_R' : \textit{Row} \rightarrow \textit{Row} \rightarrow K \\
      Γ ⊢ C_T = C_T' : \textit{Type} \rightarrow K \rightarrow K }
    { Γ ⊢ \kw{Typerec}~C~(C_B, C_I, C_S, C_L, C_R, C_T) = \kw{Typerec}~C'~(C_B', C_I', C_S', C_L', C_R', C_T') : K }
  \end{mathpar}
  \caption{Constructor and row constructor equivalence.}
  \label{fig:app:constructor-equivalence}
\end{figure*}

\begin{figure*}
  \footnotesize
  \begin{mathpar}
    \infer
    { Γ \text{ well-formed} }
    { Γ ⊢ \texttt{T}(\texttt{Bool}^*) = \texttt{Bool} : \textit{Type} }

    \infer
    { Γ \text{ well-formed} }
    { Γ ⊢ \texttt{T}(\inttyc) = \intty : \textit{Type} }

    \infer
    { Γ \text{ well-formed} }
    { Γ ⊢ \texttt{T}(\texttt{String}^*) = \texttt{String} : \textit{Type} }

    \infer
    { Γ ⊢ C : \textit{Type}}
    { Γ ⊢ \texttt{T}(\texttt{List}^*\ C) = \texttt{List}\ \texttt{T}(C) : \textit{Type} }

    \infer
    { Γ ⊢ S: \textit{Row} }
    { Γ ⊢ \texttt{T}(\texttt{Record}^*\ S) = \texttt{Record}\ \texttt{T}(S) : \textit{Type} }

    \infer
    { Γ ⊢ C : \textit{Type} }
    { Γ ⊢ \texttt{T}(\texttt{Trace}^*\ C) = \texttt{Trace}\ \texttt{T}(C) : \textit{Type} }

    \infer
    {  }
    { Γ ⊢ T(\cdot) = \cdot : \textit{Row} }

    \infer
    { Γ ⊢ C : \textit{Type} \\ Γ ⊢ S: \textit{Row} }
    { Γ ⊢ T(l:C; S) = (l:T(C); T(S)) : \textit{Row} }

    \infer
    { Γ ⊢ C = D : \textit{Type} }
    { Γ ⊢ \texttt{T}(C) = \texttt{T}(D) : \textit{Type} }

    \infer
    { Γ ⊢ S = S' : \textit{Row} }
    { Γ ⊢ \texttt{T}(S) = \texttt{T}(S') : \textit{Row} }

    \infer
    { Γ ⊢ A = B : \textit{Type}}
    { Γ ⊢ \texttt{List}\ A = \texttt{List}\ B: \textit{Type} }

    \infer
    { Γ ⊢ R = R' : \textit{Row} }
    { Γ ⊢ \texttt{Record}\ R = \texttt{Record}\ R' : \textit{Type} }

    \infer
    { Γ ⊢ A = B : \textit{Type} }
    { Γ ⊢ \texttt{Trace}\ A = \texttt{Trace}\ B : \textit{Type} }

    \infer
    { Γ ⊢ A = A' : \textit{Type} \\ Γ ⊢ B = B' : \textit{Type} }
    { Γ ⊢ A \rightarrow B = A' \rightarrow B' : \textit{Type} }

    \infer
    { Γ ⊢ A = B : \textit{Type}}
    { Γ ⊢ \forall \alpha. A = \forall \alpha. B : \textit{Type} }

    \infer
    { Γ \text{ well-formed} }
    { Γ ⊢ \cdot = \cdot : \textit{Row} }

    \infer
    { Γ ⊢ A = B: \textit{Type} \\ Γ ⊢ R = R': \textit{Row} }
    { Γ ⊢ (l:A; R) = (l:B; R') : \textit{Row} }

  \end{mathpar}
  \caption{Type and row type equivalence.}
  \label{fig:app:type-equivalence}
\end{figure*}

\subsection{Type judgments}

\begin{itemize}
\item Figure~\ref{fig:app:term-formation} defines the typing judgment
  for most of the \TLinks constructs ($Γ ⊢ M : A$)
\item Figure~\ref{fig:app:trace-intro} defines the typing rules
  introducing and eliminating traces.
\end{itemize}

\begin{figure*}
  \footnotesize
  \begin{mathpar}
    \infer
    { \Sigma(c) = A }
    { Γ ⊢ c : A}

    \infer
    { Γ(x) = A }
    { Γ ⊢ x : A}

    \infer
    { Γ ⊢ A : \textit{Type} \\
      Γ, x : A ⊢ M: B \\ x \notin \textit{Dom}(Γ)}
    { Γ ⊢ \lambda x:A.M : A \rightarrow B }

    \infer
    { Γ ⊢ M : A \rightarrow B \\ Γ ⊢ N : A }
    { Γ ⊢ M\ N : B }

    \infer
    { Γ,\alpha: K ⊢ M: A \\ \alpha \notin \textit{Dom}(Γ) }
    { Γ ⊢ \Lambda \alpha: K.M: \forall \alpha:K.A }

    \infer
    { Γ ⊢ M: \forall \alpha:K.A \\ Γ ⊢ C: K}
    { Γ ⊢ M\ C : A[\alpha \coloneqq C]}

    \infer
    { Γ ⊢ A \\ Γ,f:A ⊢ M:A}
    { Γ ⊢ \kw{fix}~f:A.M : A }

    \infer
    { Γ ⊢ L : \texttt{Bool} \\
      Γ ⊢ M : A \\
      Γ ⊢ N : A }
    { Γ ⊢ \kw{if}~L~\kw{then}~M~\kw{else}~N : A }

    \infer
    {Γ ⊢ M : \intty \\ Γ ⊢ N : \intty }
    {Γ ⊢ M + N : \intty }

    \infer
    {Γ ⊢ M : A \\ Γ ⊢ N : A \\ Γ ⊢ A : \textit{Type} }
    {Γ ⊢ M == N : \texttt{Bool} }

    \infer
    { \cdot ⊢ R : \textit{Row} }
    {Γ ⊢ \kw{table}~n~\langle \textsf{oid}:\intty, R \rangle : \texttt{List} \langle \textsf{oid}:\intty; R \rangle}

    \infer 
    { Γ ⊢ A: \textit{Type} }
    { Γ ⊢ [] : \texttt{List}\ A }

    \infer
    { Γ ⊢ M : A }
    { Γ ⊢ [M] : \texttt{List}\ A }

    \infer
    { Γ ⊢ M : \texttt{List}\ A \\ Γ ⊢ N : \texttt{List}\ A }
    { Γ ⊢ M \concat N  : \texttt{List}\ A }

    \infer 
    { Γ ⊢ M : \texttt{List}\ A \\ Γ,x:A ⊢ N : \texttt{List}\ B }
    { Γ ⊢ \kw{for}~(x \leftarrow M)~N : \texttt{List}\ B }

    \infer
    { Γ \text{ well-formed} }
    { Γ ⊢ \langle \rangle : \texttt{Record}\ ()}

    \infer
    { Γ ⊢ M : A \\ Γ ⊢ N : \texttt{Record}\ R }
    { Γ ⊢ \langle l = M; N \rangle : \texttt{Record}\ (l : A; R)}

    \infer
    { Γ ⊢ M : \texttt{Record}\ (l : A; R) }
    { Γ ⊢ M.l : A }

    \infer
    { Γ ⊢ M : B \\ Γ ⊢ A = B }
    { Γ ⊢ M : A }

    \infer
    { Γ ⊢ M : \forall \alpha: \textit{Type}. \texttt{T}(\alpha) \rightarrow \texttt{T}(C\ \alpha) \\
      Γ ⊢ N : \texttt{T}(\texttt{Record}^*\ S)}
    { Γ ⊢ \kw{rmap}^S~M~N : \texttt{T}(\texttt{Record}^*\ (\kw{Rmap}\ C \ S))}

    \infer
    { Γ ⊢ L: \texttt{T}(C) \rightarrow \texttt{T}(C) \rightarrow \texttt{T}(C) \\
      Γ ⊢ M : \texttt{T}(C) \\
      Γ ⊢ N : \texttt{T}(\texttt{Record}^*\ (\kw{Rmap}\ (\lambda\alpha.\alpha \rightarrow C)\ S)) }
    { Γ ⊢ \kw{rfold}^S\ L\ M\ N:\texttt{T}(C) }

    \infer
    { Γ ⊢ C : \textit{Type} \\
      Γ, \alpha: \textit{Type} ⊢ B : \textit{Type} \\
      \beta, \rho, \gamma \notin \textit{Dom}(Γ) \\
      Γ ⊢ M_B : B[\alpha \coloneqq \texttt{Bool}^*] \\
      Γ ⊢ M_I : B[\alpha \coloneqq \inttyc] \\
      Γ ⊢ M_S : B[\alpha \coloneqq \texttt{String}^*] \\
      Γ, \beta : \textit{Type} ⊢ M_L : B[\alpha \coloneqq \texttt{List}^*\ \beta] \\
      Γ, \rho : \textit{Row} ⊢ M_R : B[\alpha \coloneqq \texttt{Record}^*\ \rho] \\
      Γ, \gamma : \textit{Type} ⊢ M_T : B[\alpha \coloneqq \texttt{Trace}^*\ \gamma]
    }
    { Γ ⊢ \kw{typecase}~C~\kw{of}\ (M_B, M_I, M_S, \beta.M_L, \rho.M_R, \gamma.M_T) : B[\alpha \coloneqq C] }
  \end{mathpar}
  \caption{Term formation $Γ ⊢ M : A$.}
  \label{fig:app:term-formation}
\end{figure*}

\begin{figure*}
  \footnotesize
  \begin{mathpar}
    \infer
    { Γ ⊢ c : \texttt{Bool} }
    { Γ ⊢ \texttt{Lit}\ c : \texttt{Trace}\ \texttt{Bool} }

    \infer
    { Γ ⊢ c : \intty }
    { Γ ⊢ \texttt{Lit}\ c : \texttt{Trace}\ \intty }

    \infer
    { Γ ⊢ c : \texttt{String} }
    { Γ ⊢ \texttt{Lit}\ c : \texttt{Trace}\ \texttt{String} }

    \infer
    { Γ ⊢ M : \langle \mathsf{cond} : \texttt{Trace}\ \texttt{Bool}, \mathsf{out} : \texttt{Trace}\ A \rangle  }
    { Γ ⊢ \texttt{If}\ M : \texttt{Trace}\ A }

    \infer
    { Γ ⊢ C : \textit{Type} \\ Γ ⊢ M : \langle \mathsf{in} : \texttt{T(TRACE}\ C\texttt{)}, \mathsf{out} : \texttt{Trace}\ A \rangle  }
    { Γ ⊢ \texttt{For}\ C\ M : \texttt{Trace}\ A }

    \infer
    { Γ ⊢ M : \langle \mathsf{table} : \texttt{String}, \mathsf{column} : \texttt{String}, \mathsf{row} : \intty, \mathsf{data} : A \rangle  }
    { Γ ⊢ \texttt{Cell}\ M : \texttt{Trace}\ A }

    \infer
    { Γ ⊢ C : \textit{Type} \\ Γ ⊢ M : \langle \mathsf{left} : \texttt{T(TRACE}\ C\texttt{)}, \mathsf{right} : \texttt{T(TRACE}\ C\texttt{)} \rangle}
    { Γ ⊢ \texttt{OpEq}\ C\ M : \texttt{Trace}\ \texttt{Bool} }

    \infer
    { Γ ⊢ M : \langle \mathsf{left} : \texttt{Trace Int}, \mathsf{right} : \texttt{Trace Int}\rangle}
    { Γ ⊢ \texttt{OpPlus}\ M : \texttt{Trace}\ \intty }

    \infer
    { Γ ⊢ M : \texttt{Trace}\ A \\
      Γ, x_L : A ⊢ M_L : B \\
      Γ, x_I : \langle\mathsf{cond}: \texttt{Trace}\ \texttt{Bool}, \mathsf{then}: \texttt{Trace}\ A\rangle ⊢ M_I : B \\
      Γ, \alpha_F: \textit{Type}, x_F : \langle\mathsf{in}: \texttt{T}(\texttt{TRACE}\ \alpha_F), \mathsf{out}: \texttt{Trace}\ A\rangle ⊢ M_F : B \\
      Γ, x_C : \langle\mathsf{table}: \texttt{String}, \mathsf{column}: \texttt{String}, \mathsf{row}: \intty, \mathsf{data} : A\rangle ⊢ M_C : B \\
      Γ, \alpha_E : \textit{Type}, x_E : \langle\mathsf{left}: \texttt{T}(\texttt{TRACE}\ \alpha_E), \mathsf{right}: \texttt{T}(\texttt{TRACE}\ \alpha_E)\rangle ⊢ M_E : B \\
      Γ, x_P : \langle\mathsf{left}: \texttt{Trace}\ \intty, \mathsf{right}: \texttt{Trace}\ \intty\rangle ⊢ M_P : B \\
    }
    { Γ ⊢ \kw{tracecase}~M~\kw{of}\ (x_L.M_L, x_I.M_I, \alpha_F.x_F.M_F, x_C.M_C, \alpha_E.x_E.M_E, x_P.M_P) : B }
  \end{mathpar}

  \caption{\texttt{Trace} introduction and elimination rules.}
  \label{fig:app:trace-intro}
\end{figure*}

\subsection{Normalization}

\begin{itemize}
\item Figure~\ref{fig:app:normalization-beta} defines the main
  computational rules ($\beta$-rules) for normalization ($M \leadsto M'$)
\item Figure~\ref{fig:app:normalization-cc} defines commuting
  conversion rules for normalization  ($M \leadsto M'$)
\item Figure~\ref{fig:app:normalization-rules-congruence} defines
  congruence rules for normalization   ($M \leadsto M'$)
\end{itemize}

\begin{figure*}
  \footnotesize
  \begin{align*}
    (\lambda x.M)\ N &\leadsto M[x \coloneqq N] \\
    \kw{fix}~f.M &\leadsto M[f \coloneqq \kw{fix}~f.M] \\
    (\Lambda \alpha.M)\ C &\leadsto M[\alpha \coloneqq C] \\
    \kw{if}~\kw{true}~\kw{then}~M~\kw{else}~N &\leadsto M \\
    \kw{if}~\kw{false}~\kw{then}~M~\kw{else}~N &\leadsto N \\
    \langle \overline{l_i=M_i}  \rangle.l_i &\leadsto M_i \\
    \kw{rmap}^{(\overline{l_i:C_i})}\ M\ N &\leadsto \langle \overline{l_i = (M\ C_i)\ N.l_i } \rangle \\
    \kw{rfold}^{(\overline{l_i: C_i})}\ L\ M\ N & \leadsto L\ N.l_1\ (L\ N.l_2 \hdots (L\ N.l_n\ M) \hdots) \\
    \kw{for}~(x \leftarrow [])~N &\leadsto [] \\
    \kw{for}~(x \leftarrow [M])~N &\leadsto N[x \coloneqq M]
\\    \kw{tracecase}~\texttt{Lit}~M~\kw{of}\ (x.M_L, M_I, M_F, M_C, M_E, M_P)
                     &\leadsto M_L[x \coloneqq M] \\
    \kw{tracecase}~\texttt{If}~M~\kw{of}\ (M_L, x.M_I, M_F, M_C, M_E, M_P)
                     &\leadsto M_I[x \coloneqq M] \\
    \kw{tracecase}~\texttt{For}~C~M~\kw{of}\ (x.M_L, x.M_I, \alpha.x.M_F, x.M_C, \alpha.x.M_E, x.M_P)
                     &\leadsto M_F[\alpha \coloneqq C, x \coloneqq M] \\
    \kw{tracecase}~\texttt{Cell}~M~\kw{of}\ (x.M_L, x.M_I, \alpha.x.M_F, x.M_C, \alpha.x.M_E, x.M_P)
                     &\leadsto M_C[x \coloneqq M] \\
    \kw{tracecase}~\texttt{OpEq}~C~M~\kw{of}\ (x.M_L, x.M_I, \alpha.x.M_F, x.M_C, \alpha.x.M_E, x.M_P)
                     &\leadsto M_E[\alpha \coloneqq C, x \coloneqq M] \\
    \kw{tracecase}~\texttt{OpPlus}~M~\kw{of}\ (x.M_L, x.M_I, \alpha.x.M_F, x.M_C, \alpha.x.M_E, x.M_P)
                                           &\leadsto M_P[x \coloneqq M]\\
    \kw{typecase}~\texttt{Bool}~\kw{of}\ (M_B, M_I, M_S, \beta.M_L, \rho.M_R, \gamma.M_T)
                     &\leadsto M_B  \\
    \kw{typecase}~\intty~\kw{of}\ (M_B, M_I, M_S, \beta.M_L, \rho.M_R, \gamma.M_T)
                     &\leadsto M_I  \\
    \kw{typecase}~\texttt{String}~\kw{of}\ (M_B, M_I, M_S, \beta.M_L, \rho.M_R, \gamma.M_T)
                     &\leadsto M_S  \\
    \kw{typecase}~\texttt{List}~C~\kw{of}\ (M_B, M_I, M_S, \beta.M_L, \rho.M_R, \gamma.M_T)
                     &\leadsto M_L[\beta \coloneqq C] \\
    \kw{typecase}~\texttt{Record}~S~\kw{of}\ (M_B, M_I, M_S, \beta.M_L, \rho.M_R, \gamma.M_T)
                     &\leadsto M_R[\rho \coloneqq S] \\
    \kw{typecase}~\texttt{Trace}~C~\kw{of}\ (M_B, M_I, M_S, \beta.M_L, \rho.M_R, \gamma.M_T)
                     &\leadsto M_T[\gamma \coloneqq C] \\
  \end{align*}

  \caption{Normalization $\beta$-rules. See also commuting conversions in Figure~\ref{fig:app:normalization-cc}, congruence rules in Figure \ref{fig:app:normalization-rules-congruence}, and constructor computation rules in Figure~\ref{fig:app:constructor-computation} }
  \label{fig:app:normalization-beta}
\end{figure*}

\begin{figure*}
  \footnotesize
  \begin{align*}
    (\kw{if}~L~\kw{then}~M_1~\kw{else}~M_2)\ N &\leadsto \kw{if}~L~\kw{then}~M_1\ N~\kw{else}~M_2\ N \\
    (\kw{if}~L~\kw{then}~M_1~\kw{else}~M_2)\ C &\leadsto \kw{if}~L~\kw{then}~M_1\ C~\kw{else}~M_2\ C \\
    (\kw{if}~L~\kw{then}~M~\kw{else}~N).l &\leadsto \kw{if}~L~\kw{then}~M.l~\kw{else}~N.l \\
    \kw{for}~(x \leftarrow M_1 \concat M_2)~N &\leadsto (\kw{for}~(x \leftarrow M_1)~N) \concat (\kw{for}~(x \leftarrow M_2)~N) \\
    \kw{for}~(x \leftarrow \kw{for}~(y \leftarrow L)~M)~N &\leadsto \kw{for}~(y \leftarrow L)~\kw{for}~(x \leftarrow M)~N
\\
                \kw{if}~(\kw{if}~L~\kw{then}~M_1~\kw{else}~M_2)~\kw{then}~N_1~\kw{else}~N_2 
                &\leadsto \kw{if}~L~\kw{then}~(\kw{if}~M_1~\kw{then}~N_1~\kw{else}~N_2)~\kw{else}~(\kw{if}~M_2~\kw{then}~N_1~\kw{else}~N_2) \\[.5\baselineskip]
                \kw{for}~(x \leftarrow \kw{if}~L~\kw{then}~M_1~\kw{else}~M_2)~N 
                & \leadsto \kw{if}~L~\kw{then}~\kw{for}~(x \leftarrow M_1)~N~\kw{else}~\kw{for}~(x \leftarrow M_2)~N \\[.5\baselineskip]
                \kw{tracecase}~\kw{if}~L~\kw{then}~M_1~\kw{else}~M_2~\kw{of}\ (M_L, M_I, M_F, M_C, M_E, M_P) 
                & \leadsto \kw{if}~L~\kw{then}~\kw{tracecase}~M_1~\kw{of}\ (M_L, M_I, M_F, M_C, M_E, M_P)\\
                & \qquad \quad \kw{else}~\kw{tracecase}~M_2~\kw{of}\ (M_L, M_I, M_F, M_C, M_E, M_P)
  \end{align*}

  \caption{Commuting conversions reorder expressions to expose more $\beta$-reductions.}
  \label{fig:app:normalization-cc}
\end{figure*}

\begin{figure*}
\footnotesize
\centering
  \begin{mathpar}
    \infer
    {M \rightsquigarrow M'}
    {I[M] \rightsquigarrow I[M']}

    \infer
    {C \rightsquigarrow C'}
    {J[C] \rightsquigarrow J[C']}
  \end{mathpar}\\
  \[
    \begin{array}{lrcl}
      \text{Term frames} & I[] & \Coloneqq
      & \lambda x.[] \mid
        []\ N \mid M\ [] \mid
        \Lambda \alpha.[] \mid
        []\ C \mid
        \kw{if}~[]~\kw{then}~M~\kw{else}~N \mid
        \kw{if}~L~\kw{then}~[]~\kw{else}~N
        \kw{if}~L~\kw{then}~M~\kw{else}~[]
      \\ & & \mid &
        \langle l = []; N \rangle \mid
        \langle l = M; [] \rangle \mid
        [].l \mid
        \kw{rmap}^S\ []\ N \mid
        \kw{rmap}^S\ M\ [] \mid
        \kw{rfold}^S\ []\ M\ N \mid
        \kw{rfold}^S\ L\ []\ N \mid
        \kw{rfold}^S\ L\ M\ []
      \\ & & \mid &
        \texttt{[} [] \texttt{]} \mid
        [] \concat N \mid
        M \concat N \mid
        [] == N \mid
        M == N \mid
        [] + N \mid
        M + N \mid
                   \kw{for}~(x \leftarrow [])~N \mid
        \kw{for}~(x \leftarrow M)~[] 
      \\ & & \mid &
        \texttt{Lit}\ [] \mid
        \texttt{If}\ [] \mid
        \texttt{For}\ C\ [] \mid
                 \texttt{Cell}\ [] \mid
        \texttt{OpEq}\ C\ [] \mid
        \texttt{OpPlus}\ [] \\ & &  \mid &
        \kw{tracecase}~[]~\kw{of}\ (M_L, M_I, M_F, M_C, M_E, M_P) \mid
        \kw{tracecase}~M~\kw{of}\ ([], M_I, M_F, M_C, M_E, M_P) \\ & & \mid &
        \kw{tracecase}~M~\kw{of}\ (M_L, [], M_F, M_C, M_E, M_P) \mid
        \kw{tracecase}~M~\kw{of}\ (M_L, M_I, [], M_C, M_E, M_P) \\ & & \mid &
        \kw{tracecase}~M~\kw{of}\ (M_L, M_I, M_F, [], M_E, M_P) \mid
        \kw{tracecase}~M~\kw{of}\ (M_L, M_I, M_F, M_C, [], M_P) \\ & & \mid &
        \kw{tracecase}~M~\kw{of}\ (M_L, M_I, M_F, M_C, M_E, []) \\ & & \mid &
        \kw{typecase}~C~\kw{of}\ ([], M_I, M_S, \beta.M_L, \rho.M_R, \gamma.M_T) \\ & & \mid &
        \kw{typecase}~C~\kw{of}\ (M_B, [], M_S, \beta.M_L, \rho.M_R, \gamma.M_T) \mid
        \kw{typecase}~C~\kw{of}\ (M_B, M_I, [], \beta.M_L, \rho.M_R, \gamma.M_T) \\ & & \mid &
        \kw{typecase}~C~\kw{of}\ (M_B, M_I, M_S, \beta.[], \rho.M_R, \gamma.M_T) \mid
        \kw{typecase}~C~\kw{of}\ (M_B, M_I, M_S, \beta.M_L, \rho.[], \gamma.M_T) \\ & & \mid &
        \kw{typecase}~C~\kw{of}\ (M_B, M_I, M_S, \beta.M_L, \rho.M_R, \gamma.[]) \\
      \\
      \text{Constructor frames} & J[] & \Coloneqq
      & M\ [] \mid
        \kw{rmap}^{[]}\ M\ N \mid
        \kw{rfold}^{[]}\ L\ M\ N \mid
        \texttt{For}\ []\ M \mid
                 \texttt{OpEq}\ []\ M \\
      & & \mid &
        \kw{typecase}~[]~\kw{of}\ (M_B, M_I, M_S, \beta.M_L, \rho.M_R, \gamma.M_T)
    \end{array}
  \]
  \caption{Congruence rules allow subterms to reduce independently.}
  \label{fig:app:normalization-rules-congruence}
\end{figure*}

\section{Proofs}

\subsection{Additional properties}\label{sec:app:additional}

Besides the properties stated in the main body of the paper, the
following additional properties are needed:

\begin{lemma}[Substitution lemmas]\label{lem:substitution-lemmas}
  ~
  \begin{enumerate}
  \item If $\Gamma, x:A \vdash M: B$ and $\Gamma \vdash N:A$ then $\Gamma \vdash M[x\coloneqq N]:B$.
  \item If $\Gamma, \alpha:K \vdash A: K'$ and $\Gamma \vdash C:K$ then $\Gamma[\alpha \coloneqq C] \vdash A[\alpha \coloneqq C]:K'[\alpha \coloneqq C]$.
  \item If $\Gamma, \rho:K \vdash A: K'$ and $\Gamma \vdash S:K$ then $\Gamma[\rho \coloneqq S] \vdash A[\rho \coloneqq S]:K'[\rho \coloneqq S]$.
  \item If $\Gamma, \alpha:K \vdash M: A$ and $\Gamma \vdash C:K$ then $\Gamma[\alpha \coloneqq C] \vdash M[\alpha \coloneqq C]:A[\alpha \coloneqq C]$.
  \item If $\Gamma, \rho:K \vdash M: A$ and $\Gamma \vdash S:K$ then $\Gamma[\rho \coloneqq S] \vdash M[\rho \coloneqq S]:A[\rho \coloneqq S]$.
  \end{enumerate}
\end{lemma}

\begin{lemma}[Weakening]\label{lem:weakening}
  If $\Gamma \vdash M : A$, $\Gamma \vdash B : K$, and $x$ does not appear free in $\Gamma$, $M$, $A$, then $\Gamma, x:B \vdash M:A$.
\end{lemma}

\begin{lemma}[Context swap]\label{lem:context-swap}~
  \begin{enumerate}
  \item If $\Gamma, x:A_x, y:A_y \vdash M : B$ then $\Gamma, y:A_y, x:A_x \vdash M : B$.
  \item If $\Gamma, x:A_x, y:A_y \vdash B : K_B$ then $\Gamma, y:A_y, x:A_x \vdash B : K_B$.
  \item If $\Gamma, \alpha:K_\alpha, y:A_y \vdash M : B$ and $\alpha$ does not appear free in $A_y$ then $\Gamma, y:A_y, \alpha:K_\alpha \vdash M : B$.
  \item If $\Gamma, \alpha:K_\alpha, y:A_y \vdash B : K_B$ and $\alpha$ does not appear free in $A_y$ then $\Gamma, y:A_y, \alpha:K_\alpha \vdash B : K_B$.
  \item If $\Gamma, x:A_x, \beta:K_\beta \vdash M : B$ then $\Gamma, \beta:K_\beta, x:A_x,\vdash M : B$.
  \item If $\Gamma, x:A_x, \beta:K_\beta \vdash B : K_B$ then $\Gamma, \beta:K_\beta, x:A_x,\vdash B : K_B$.
  \item If $\Gamma, \alpha:K_\alpha \beta:K_\beta \vdash M : B$ and $\alpha$ does not appear free in $K_\beta$ then $\Gamma, \beta:K_\beta, \alpha:K_\alpha\vdash M : B$.
  \item If $\Gamma, \alpha:K_\alpha \beta:K_\beta \vdash B : K_B$ and $\alpha$ does not appear free in $K_\beta$ then $\Gamma, \beta:K_\beta, \alpha:K_\alpha\vdash B : K_B$.
  \end{enumerate}
\end{lemma}

\begin{lemma}\label{lem:valuetf-inverse}
  For all query type constructors $C$ and row constructors $S$ and well-formed contexts $\Gamma$:
  \[ \Gamma \vdash \mathtt{VALUE}(\mathtt{TRACE}\ C) = C \] and
  \[ \Gamma \vdash \kw{Rmap}\ \mathtt{VALUE}\ (\kw{Rmap}\ \mathtt{TRACE}\ S) = S \]
\end{lemma}

\begin{lemma}\label{lem:trace-non-base}
  For all query types $C$, \lstinline|TRACE| $C$ is not a base type.
\end{lemma}

\begin{definition}[Trace context]\label{def:trace-context}
  $\trace{\Gamma}$ maps term variable $x$ to \lstinline|T(TRACE $C$)| if and only if $\Gamma$ maps $x$ to $A$, where $C$ is the obvious constructor with $\cdot \vdash A = \texttt{T}(C)$.
\end{definition}

\begin{lemma}\label{lem:nrc-types-constr}
  For every query type $A$ made of base types, list constructors, and closed records, there exists $C$ such that $\Gamma \vdash A = \mathtt{T}(C)$ in a well-formed context $\Gamma$.
\end{lemma}

\subsection{Proof of Lemma~\ref{lem:valuetf-inverse}}
\begin{proof}
  By induction on query types $C$ and closed rows of query types $S$.
  \begin{itemize}
  \item Base types \texttt{Bool}$^*$, \intty$^*$, \texttt{String}$^*$:
    \[ \mathtt{VALUE}(\mathtt{TRACE}\ \mathtt{Bool}^* )
      = \mathtt{VALUE}(\mathtt{Trace}\ \mathtt{Bool}^* )
      = \mathtt{Bool}^*
    \]
  \item List types \texttt{List}$^*$ $D$:
    \begin{align*}
      \mathtt{VALUE}(\mathtt{TRACE}\ (\mathtt{List}^*\ D))
      &= \mathtt{VALUE}(\mathtt{List}^*\ (\mathtt{TRACE}\ D)) \\
      &= \mathtt{List}^*\ (\mathtt{VALUE}(\mathtt{TRACE}\ D)) \\
      &= \mathtt{List}^*\ D
    \end{align*}
  \item Record types \texttt{Record}$^*\ S$:
    \begin{align*}
      \mathtt{VALUE}(\mathtt{TRACE}\ (\mathtt{Record}^*\ S))
      &= \mathtt{VALUE}(\mathtt{Record}^* (\kw{Rmap}\ \mathtt{TRACE}\ S)) \\
      &= \mathtt{Record}^*\ (\kw{Rmap}\ \mathtt{VALUE}\ (\kw{Rmap}\ \mathtt{TRACE}\ S)) \\
      &= \mathtt{Record}^*\ S
    \end{align*}
  \item Empty row $\cdot$: $\kw{Rmap}\ \mathtt{VALUE}\ (\kw{Rmap}\ \mathtt{TRACE}\ \cdot) = \cdot$
  \item Row cons $(l:A, S)$:
    \begin{align*}
      & \kw{Rmap}\ \mathtt{VALUE}\ (\kw{Rmap}\ \mathtt{TRACE}\ (l:A, S)) \\
      =& \kw{Rmap}\ \mathtt{VALUE}\ (l:\mathtt{TRACE}\ A, \kw{Rmap}\ \mathtt{TRACE}\ S) \\
      =& (l:\mathtt{VALUE}\ (\mathtt{TRACE}\ A), \kw{Rmap}\ \mathtt{VALUE}\ (\kw{Rmap}\ \mathtt{TRACE}\ S)) \\
      =& (l:A, \kw{Rmap}\ \mathtt{VALUE}\ (\kw{Rmap}\ \mathtt{TRACE}\ S)) \\
      =& (l:A, S)
    \end{align*}\qedhere
  \end{itemize}
\end{proof}

\subsection{Proof of Lemma~\ref{lem:trace-non-base}}
\begin{proof}
  By induction on query types $C$ made up from base types, lists, and closed records.
  Applying \lstinline|TRACE| to base types \lstinline|Bool|, \lstinline|Int|, and \lstinline|String| results in traced base types \lstinline|Trace Bool|, \lstinline|Trace Int|, and \lstinline|Trace String|, respectively.
  List types are guarded by the \lstinline|List| type constructor, and similarly for records.
  Traces are not query types, but if they were, the induction hypothesis would apply.
\end{proof}

\subsection{Proof of Lemma~\ref{lem:dist-type-correctness}}
\begin{proof}
  By induction on the query type $C$.

  \begin{itemize}
  \item The base cases are $\texttt{Bool}$, $\texttt{Int}$, and $\texttt{String}$.
    For any base type $O$ out of these, we have $\texttt{TRACE}\ O = \texttt{Trace}\ O$.
    We have $\mathit{dist}(\texttt{Trace}\ O, k, t) = k[\mathbb{H} \coloneqq t]$ and need to show that it has type $\texttt{Trace}\ O$.
    Both $t$ and $\mathbb{H}$ have type $\texttt{Trace}\ O$, so substituting one for the other in $k$ does not change the type (Lemma \ref{lem:substitution-lemmas}).

  \item Case $C = \texttt{List}\ (\texttt{TRACE}\ C')$:
    We need the right-hand side $\kw{for}\ (x \leftarrow l)\ [\mathit{dist}(\texttt{TRACE}\ C', k, x)] $ to have type $\texttt{TRACE}\ (\texttt{List}\ C')$.
    We use the rules for comprehension and singleton list.
    We now need to show that $\mathit{dist}(\texttt{TRACE}\ C', k, x)$ has type $\texttt{TRACE}\ C'$ which is true by induction hypothesis with the same $k$.

  \item Case $C = \langle \overline{l:\texttt{TRACE}\ C'} \rangle$:
    The right-hand side $\langle \overline{ l = \mathit{dist}(\texttt{TRACE}\ C',k,r.l) }\rangle$ needs to have type $\langle \overline{l:\texttt{TRACE}\ C'} \rangle$.
    Thus, by record construction and record projection, we need each of the expressions $\mathit{dist}(\texttt{TRACE}\ C',k,r.l)$ to have type $\texttt{TRACE}\ C'$ which they do by induction hypothesis.
    \qedhere
  \end{itemize}
\end{proof}

\subsection{Proof of Theorem~\ref{thm:trace-type-correctness-2}}
\begin{proof}
  By induction on the typing derivation for $M : \texttt{T}(C)$.
  Almost all cases require that some subterms have a type $\mathtt{T}(C')$ that is equal to some query type $A$.
  We can obtain this constructor $C'$ by Lemma~\ref{lem:nrc-types-constr}.

  \begin{itemize}
  \item Case $\infer { Γ(x) = A } { Γ ⊢ x : A }$:
    $
      \infer
      {\trace{\Gamma}(x) = \mathtt{T}(\mathtt{TRACE}\ C) \quad \text{(Definition \ref{def:trace-context})}}
      {\trace{\Gamma} \vdash x : \mathtt{T}(\mathtt{TRACE}\ C)}
    $

  \item Literals $c$ have base types \lstinline|Bool|, \lstinline|Int|, or \lstinline|String|.
    Their traces \lstinline|Lit| $c$ have types \lstinline|Trace Bool|, \lstinline|Trace Int|, or \lstinline|Trace String|, respectively.

  \item Case $ \infer
    { Γ ⊢ L : \texttt{Bool} \\
      Γ ⊢ M : A \\
      Γ ⊢ N : A }
    { Γ ⊢ \kw{if}~L~\kw{then}~M~\kw{else}~N : A }$:

    The right hand side of the self-tracing transform is another \kw{if-then-else} with condition $\texttt{value (Trace Bool)}\ \trace{L}$ and then-branch
    \[
      \mathit{dist}(\texttt{TRACE}\ C, \texttt{If}\ \langle \texttt{cond} = \trace{L}, \texttt{out} = \mathbb{H} \rangle, \trace{M})
    \] and similar else-branch.

    In the condition, we apply $\texttt{value}: \forall\alpha.\texttt{T}(\alpha) \rightarrow \texttt{T}(\texttt{VALUE}\ \alpha)$ to a subtrace of type \lstinline|TRACE Bool| by induction hypothesis.
    Therefore it has type \lstinline|VALUE (TRACE Bool)| which is equal to \lstinline|Bool| by Lemma~\ref{lem:valuetf-inverse}.

    For all base types $D$, $\texttt{If}\ \langle \texttt{cond} = \trace{L}, \texttt{out} = \mathbb{H} \rangle$ has type $\texttt{Trace}\ D$ assuming $\mathbb{H}:\texttt{Trace}\ D$.
    We have $\trace{M}: \mathtt{T}(\mathtt{TRACE}\ C)$ by IH.
    Therefore, by Lemma~\ref{lem:dist-type-correctness}, the whole term obtained by \textit{dist} has type $\texttt{TRACE}\ C$.
    The else-branch is analogous and the whole expression has type $\mathtt{T}(\mathtt{TRACE}\ C)$.

  \item Case $ \infer { } { Γ ⊢ [] : \texttt{List}\ A } $:
    \[
      \infer*
      {\infer*
        {\infer*
          { \trace{Γ} ⊢ \mathtt{T}(\mathtt{TRACE}\ C) : \mathit{Type} \text{ using } A = \mathtt{T}(C) }
          {
            \trace{Γ} ⊢ [] : \texttt{List}\ \mathtt{T}(\mathtt{TRACE}\ C)
          }
        }
        {\trace{Γ} ⊢ [] : \mathtt{T}(\texttt{List}^*\ (\mathtt{TRACE}\ C))}
      }
      {\trace{Γ} ⊢ [] : \mathtt{T}(\mathtt{TRACE}\ (\texttt{List}^*\ C))}
    \]

  \item Case $ \infer { Γ ⊢ M : A } { Γ ⊢ [M] : \texttt{List}\ A } $:
    \[
      \infer*
      { \infer*
        { \infer*
          { \text{IH} }
          { \trace{Γ} ⊢ \trace{M} : \mathtt{T}(\mathtt{TRACE}\ C)}
        }
        {
          \trace{Γ} ⊢ [\trace{M}]: \texttt{List}\ \mathtt{T}(\mathtt{TRACE}\ C)
        }
      }
      {
        \trace{Γ} ⊢ [\trace{M}]: \mathtt{T}(\mathtt{TRACE}\ (\texttt{List}^*\ C))
      }
    \]

  \item Case $ \infer
    { Γ ⊢ M : \texttt{List}\ A \and Γ ⊢ N : \texttt{List}\ A }
    { Γ ⊢ M \concat N  : \texttt{List}\ A }$:

    \[
      \infer*
      {
        \infer*
        {
          \infer*
          {
            \infer*
            { \text{IH} }
            { \trace{Γ} ⊢ \trace{M} : \mathtt{T}(\mathtt{TRACE}\ (\texttt{List}^*\ C)) }
          }
          { \trace{Γ} ⊢ \trace{M} : \texttt{List}\ \mathtt{T}(\mathtt{TRACE}\ C) } \and \text{ analogous for $N$}
        }
        { \trace{Γ} ⊢ \trace{M} \concat \trace{N} : \texttt{List}\ \mathtt{T}(\mathtt{TRACE}\ C) }
      }
      { \trace{Γ} ⊢ \trace{M} \concat \trace{N} : \mathtt{T}(\mathtt{TRACE}\ (\texttt{List}^*\ C)) }
    \]

  \item Case $ \infer
    { Γ ⊢ M : \texttt{List}\ B \\ Γ,x:B ⊢ N : \texttt{List}\ A }
    { Γ ⊢ \kw{for}~(x \leftarrow M)~N : \texttt{List}\ A } $:

    \[\footnotesize
      \infer*
      {
        \infer*
        { \infer*
          { \infer*
            { \text{IH} }
            { \trace{Γ} ⊢ \trace{M} : \mathtt{T}(\mathtt{TRACE}\ (\texttt{List}^*\ D)) }
          }
          { \trace{Γ} ⊢ \trace{M} : \mathtt{List}\ \mathtt{T}(\mathtt{TRACE}\ D) }
          \and
          \infer* { \star } {\trace{Γ},x:\mathtt{T}(\mathtt{TRACE}\ D) ⊢ b : \mathtt{List}\ \mathtt{T}(\mathtt{TRACE}\ C)} }
        { \trace{Γ} ⊢ \kw{for}~(x \leftarrow \trace{M})\ b : \texttt{List}\ \mathtt{T}(\mathtt{TRACE}\ C) }
      }
      { \trace{Γ} ⊢ \kw{for}~(x \leftarrow \trace{M})\ b : \mathtt{T}(\mathtt{TRACE}\ (\texttt{List}^*\ C)) }
    \]
    where $b = \mathit{dist}(\texttt{TRACE}\ C, \texttt{For}\ D\ \langle \texttt{in} = x, \texttt{out} = \mathbb{H} \rangle, \trace{N})$ and
    $\star$ follows from the induction hypothesis applied to $\trace{N}$ and Lemma~\ref{lem:dist-type-correctness}.

  \item The case for records is similar to that for list concatenation, in that we have multiple subtraces where the induction hypothesis applies, we just collect them into a record instead of another list concatenation.

  \item Case record projection: The projection was well-typed before tracing, so the record term $M$ contains label $l$ with some type $A$.
    By induction hypothesis and $A = \mathtt{T}(\mathtt{TRACE}\ C)$ the trace of $M$ contains label $l$ with type \lstinline|TRACE $C$|.
    \[
        \infer
        { \infer
          { \text{IH} }
          { \trace{Γ} ⊢ \trace{M} : \langle l^\bullet: \texttt{T(TRACE $C$)}, \dots \rangle } }
        { \trace{Γ} ⊢ \trace{M}.l : \texttt{T(TRACE $C$)} }
      \]

  \item Case \kw{table}: This is a slightly more complicated version of the base case for constants.
    We essentially map the \texttt{Cell} trace constructor over every table cell.
    Thus we go from a list of records of base types to a list of records of \texttt{Trace}d base types.

    {\small
      \[\infer
        { \infer
          { \trace{Γ} ⊢ \kw{table}~\dots \\
            \infer*
            { \infer
              { \trace{Γ}, y:\langle\overline{l :C}\rangle ⊢ y.l : C
              }
              { \star }
            }
            { \trace{Γ}, y:\langle\overline{l :C}\rangle ⊢ [\langle\overline{l = \mathit{cell}(n, l, y.\texttt{oid}, y.l)}\rangle ] : [\langle\overline{l : \texttt{Trace}~C}\rangle] }
          }
          { \trace{Γ} ⊢ \kw{for}\ (y \leftarrow \kw{table}~n~\langle \overline{l : C} \rangle)~[\langle\overline{l = \mathit{cell}(n, l, y.\texttt{oid}, y.l)}\rangle ] : [\langle\overline{l : \texttt{Trace}~C}\rangle] }
        }
        { \trace{Γ} ⊢ \kw{for}\ (y \leftarrow \kw{table}~n~\langle \overline{l : C} \rangle)~[\langle\overline{l = \mathit{cell}(n, l, y.\texttt{oid}, y.l)}\rangle ] : \texttt{T(TRACE [$\langle\overline{l :C}\rangle$])}
        }
      \]
    }

    There are a couple of steps missing at $\star$.
    The singleton list step is trivial.
    Then we have one precondition for each column in the table.
    Recall that \textit{cell} is essentially an abbreviation for \lstinline!Cell!, which records table name, column name, row number, and the actual cell data in a trace.
    We use the table name $n$ and the record label $l$ as string values for the table and column fields.
    We enforce in the typing rules that every table has the \lstinline|oid| column of type \texttt{Int}.

  \item Case equality:
    {\small
      \[
        \infer
        { \infer
          { }
          { \trace{Γ} ⊢ C : \mathit{Type}} \\
          \infer
          { \text{IH} }
          { \trace{Γ} ⊢ \trace{M} : \texttt{T(TRACE $C$)}} \\
          \infer
          { \text{IH} }
          { \trace{Γ} ⊢ \trace{N} : \texttt{T(TRACE $C$)}} \\
        }
        { \trace{Γ} ⊢ \texttt{OpEq}\ C\ \langle \texttt{left} = \trace{M}, \texttt{right} = \trace{N} \rangle : \texttt{Trace Bool} }
      \]
    }

  \item Case plus, with liberal application of \lstinline!T(TRACE Int) = Trace Int!:
    \[
      \infer
      { \infer
        { \text{Induction hypothesis} }
        { \trace{Γ} ⊢ \trace{M} : \texttt{T(TRACE Int)} } \\
        \infer
        { \text{Induction hypothesis} }
        { \trace{Γ} ⊢ \trace{N} : \texttt{T(TRACE Int)} } \\
      }
      { \trace{Γ} ⊢ \texttt{OpPlus}\ \langle \texttt{left} = \trace{M}, \texttt{right} = \trace{N} \rangle : \texttt{T(TRACE Int)} }
    \]
    \qedhere
  \end{itemize}
\end{proof}

\subsection{Proof of Lemma~\ref{lem:preservation-constructor}}\label{sec:prf:preservation-constructor}
\begin{proof}
  By induction on the kinding derivation.
  We look at the possible reductions (see Figure \ref{fig:app:constructor-computation}).
  Congruence rules allow for reduction in rows, function bodies, applications, list, trace, record, row map, and typerec. These all follow directly from the induction hypothesis.
  The remaining cases are:
  \begin{itemize}
  \item $(\lambda\alpha:K.C)\ D \leadsto C[\alpha \coloneqq D]$: by Lemma \ref{lem:substitution-lemmas}.
  \item $\kw{Rmap}\ C\ \cdot \leadsto \cdot$: both sides have kind \textit{Row}.
  \item $\kw{Rmap}\ C\ (l:D; S) \leadsto (l: C\ D; \kw{Rmap}\ C\ S)$: from the induction hypothesis we have that
    $C$ has kind $\mathit{Type} \rightarrow \mathit{Type}$,
    $D$ has kind $\mathit{Type}$, and $S$ has kind \textit{Row}.
    Therefore $C\ D$ has kind \textit{Type} and the whole right-hand side has kind \textit{Row}.
  \item Typerec $\beta$-rules:
    \begin{itemize}
    \item Base type right hand sides have kind \textit{Type} by IH.
    \item Lists:
      \[ \kw{Typerec}~\texttt{List}^*~D~(C_B, C_I, C_S, C_L, C_R, C_T) \leadsto C_L\ D\ (\kw{Typerec}~D~(C_B, C_I, C_S, C_L, C_R, C_T))\]

      $C_L$ has kind $\mathit{Type} \rightarrow K \rightarrow K$ by IH. $D$ has kind \textit{Type} by IH, and the typerec expression has kind $K$.

    \item Records:
      \[ \kw{Typerec}~\texttt{Record}^*~S~(C_B, C_I, C_S, C_L, C_R, C_T) \leadsto C_R\ S\ (\kw{Rmap}\ (\lambda \alpha.\kw{Typerec}~\alpha~(C_B, C_I, C_S, C_L, C_R, C_T))\ S) \]

      $C_L$ has kind $\mathit{Row} \rightarrow \mathit{Row} \rightarrow K$ by IH.
      $S$ has kind \textit{Row} by IH.
      The row map expression has kind \textit{Row}, because the type-level function has kind $\mathit{Type} \rightarrow \mathit{Type}$.

    \item The trace case is analogous to the list case.
      \qedhere
    \end{itemize}
  \end{itemize}
\end{proof}

\subsection{Proof of Lemma~\ref{lem:preservation}}\label{sec:prf:preservation}
\begin{proof} By induction on the typing derivation $\Gamma \vdash M : A$.
  Constants, variables, empty lists, and empty records do not reduce.
  We omit discussion of the cases that follow directly from the induction hypothesis, Lemma \ref{lem:preservation-constructor}, and congruence rules (see Figure \ref{fig:app:normalization-rules-congruence}), like $M + N$ being able to reduce in both $M$ and~$N$.
  The remaining, interesting reduction rules are the $\beta$-rules in Figure~\ref{fig:app:normalization-beta} and the commuting conversions in Figure~\ref{fig:app:normalization-cc}.
  We discuss them grouped by the relevant typing rule.
  \begin{itemize}
  \item Function application:
    \begin{itemize}
    \item $(\lambda x.M)\ N \leadsto M[x \coloneqq N]$: follows from Lemma \ref{lem:substitution-lemmas}.
    \item $(\kw{if}~L~\kw{then}~M_1~\kw{else}~M_2)\ N \leadsto \kw{if}~L~\kw{then}~M_1\ N~\kw{else}~M_2\ N$:

      We have: \[\footnotesize
        \infer*
        { \infer*
          { Γ ⊢ L : \mathtt{Bool} \and Γ ⊢ M_1 : A \rightarrow B \and Γ ⊢ M_2 : A \rightarrow B}
          { Γ ⊢ \kw{if}~L~\kw{then}~M_1~\kw{else}~M_2 : A \rightarrow B } \and Γ ⊢ N : A }
        { Γ ⊢ (\kw{if}~L~\kw{then}~M_1~\kw{else}~M_2)\ N : B }
      \]
      and can therefore show: \[\footnotesize
        \infer*
        { Γ ⊢ L : \mathtt{Bool} \and
          \infer* { Γ ⊢ M_1 : A \rightarrow B \and Γ ⊢ N : A } { Γ ⊢ M_1\ N : B } \and
          \infer* { Γ ⊢ M_2 : A \rightarrow B \and Γ ⊢ N : A } { Γ ⊢ M_2\ N : B }
        }
        { Γ ⊢ \kw{if}~L~\kw{then}~M_1\ N~\kw{else}~M_2\ N : B }
      \]
    \end{itemize}

  \item Type instantiation:
    \begin{itemize}
    \item $(\Lambda \alpha.M)\ C \leadsto M[\alpha \coloneqq C]$: follows from the constructor substitution lemma (Lemma \ref{lem:substitution-lemmas}). 
    \item $(\kw{if}~L~\kw{then}~M_1~\kw{else}~M_2)\ C \leadsto \kw{if}~L~\kw{then}~M_1\ C~\kw{else}~M_2\ C$: hoisting if-then-else out of the term works the same as application above.
    \end{itemize}

  \item Fixpoint: follows from the substitution lemma (Lemma \ref{lem:substitution-lemmas}).
  \item If-then-else: if the condition is a Boolean constant, the expression reduces to the appropriate branch, which has the correct type by IH.
    The commuting conversion for lifting if-then-else out of the condition is type-correct by IH and rearranging of if-then-else rules.
  \item List comprehensions:
    \begin{itemize}
    \item The if-then-else commuting conversion is as before.
    \item $\kw{for}~(x \leftarrow [])~N \leadsto []$: $[]$ has any list type and $N$ has a list type.
    \item $ \kw{for}~(x \leftarrow [M])~N \leadsto N[x \coloneqq M] $: by substitution (Lemma \ref{lem:substitution-lemmas}).
    \item $ \kw{for}~(x \leftarrow M_1 \concat M_2)~N \leadsto (\kw{for}~(x \leftarrow M_1)~N) \concat (\kw{for}~(x \leftarrow M_2)~N) $: reorder rules.
    \item $\kw{for}~(x \leftarrow \kw{for}~(y \leftarrow L)~M)~N \leadsto \kw{for}~(y \leftarrow L)~\kw{for}~(x \leftarrow M)~N $:

      We have:
      \[ \infer*
        { \infer*
          { Γ ⊢ L : [A_L] \\
            Γ, y:A_L ⊢ M : [A_M] }
          { Γ ⊢ \kw{for}~(y \leftarrow L)~M : [A_M]
          } \and Γ,x:A_M ⊢ N : [A_N]
        }
        { Γ ⊢ \kw{for}~(x \leftarrow \kw{for}~(y \leftarrow L)~M)~N : [A_N] }
      \]

      We need:
      \[ \infer*
        { Γ ⊢ L:[A_L] \and
          \infer*
          { Γ, y:A_L ⊢ M:[A_M] \and
            Γ, y:A_L, x:A_M ⊢ N : [A_N] }
          { Γ, y:A_L ⊢ \kw{for}~(x \leftarrow M)~N : [A_N] }
        }
        { Γ ⊢ \kw{for}~(y \leftarrow L)~\kw{for}~(x \leftarrow M)~N : [A_N] }
      \]

      We obtain $Γ, y:A_L, x:A_M ⊢ N : [A_N]$ from $Γ,x:A_M ⊢ N : [A_N]$ by weakening (Lemma \ref{lem:weakening}) and context swap (Lemma \ref{lem:context-swap}).
    \end{itemize}

  \item Projection: The $\beta$ rule is obvious, the if-then-else commuting conversion is as before.

  \item Type equality $\infer { Γ ⊢ N : B \\ Γ ⊢ A = B } { Γ ⊢ N : A }$: for all $N'$ with $N \leadsto N'$ we have that $\Gamma \vdash N' : B$ by the induction hypothesis. We also know that $Γ ⊢ A = B$, so $\Gamma \vdash N' : A$ by this typing rule and symmetry of type equality.

  \item Case \kw{rmap}:
    Typing rule:
    \[    \infer
      { Γ ⊢ M : \forall \alpha: \textit{Type}. \texttt{T}(\alpha) \rightarrow \texttt{T}(C\ \alpha) \\
        Γ ⊢ N : \texttt{T}(\texttt{Record}^*\ S)}
      { Γ ⊢ \kw{rmap}^S~M~N : \texttt{T}(\texttt{Record}^*\ (\kw{Rmap}\ C \ S))}
    \]
    Reduction rule:
    \[
      \kw{rmap}^{\langle \overline{l_i:C_i}  \rangle}\ M\ N
      \leadsto \langle \overline{l_i = (M\ C_i)\ N.l_i } \rangle
    \]
    Need to show that $\langle \overline{l_i = (M\ C_i)\ N.l_i } \rangle : \texttt{T}(\texttt{Record}^*\ (\kw{Rmap}\ C \ \langle \overline{l_i:C_i}  \rangle))$.
    By row type constructor evaluation, that type equals $\texttt{T}(\texttt{Record}^*\ \langle \overline{l_i:C\ C_i}  \rangle)$, which is the obvious type of $\langle \overline{l_i = (M\ C_i)\ N.l_i } \rangle$.

  \item Case \kw{rfold}:     Typing rule:
    \[    \infer
    { Γ ⊢ L: \texttt{T}(C) \rightarrow \texttt{T}(C) \rightarrow \texttt{T}(C) \\
      Γ ⊢ M : \texttt{T}(C) \\
      Γ ⊢ N : \texttt{T}(\texttt{Record}^*\ (\kw{Rmap}\ (\lambda\alpha.\alpha \rightarrow C)\ S)) }
    { Γ ⊢ \kw{rfold}^S\ L\ M\ N:\texttt{T}(C) }
  \]
  Reduction rule:
  \[
    \kw{rfold}^{(\overline{l_i: C_i})}\ L\ M\ N \leadsto L\ N.l_1\ (L\ N.l_2 \hdots (L\ N.l_n\ M) \hdots)
  \]
  Need to show that $L\ N.l_1\ (L\ N.l_2 \hdots (L\ N.l_n\ M) \hdots)$ has type \lstinline|T($C$)|.
  $M$ has type \lstinline|T($C$)|.
  $L$ has type $\texttt{T}(C) \rightarrow \texttt{T}(C) \rightarrow \texttt{T}(C)$.
  Each $N.l_i$ has type \lstinline|T($C$)|, because $N$ has a record type obtained by mapping the constant function with result $C$ over row $S$.

  \item Typecase typing rule:
    \[ \infer
      { Γ ⊢ C : \textit{Type} \\
        Γ, \alpha: \textit{Type} ⊢ B : \textit{Type} \\
        \beta, \rho, \gamma \notin \textit{Dom}(Γ) \\
        Γ ⊢ M_B : B[\alpha \coloneqq \texttt{Bool}^*] \\
        Γ ⊢ M_I : B[\alpha \coloneqq \inttyc] \\
        Γ ⊢ M_S : B[\alpha \coloneqq \texttt{String}^*] \\
        Γ, \beta : \textit{Type} ⊢ M_L : B[\alpha \coloneqq \texttt{List}^*\ \beta] \\
        Γ, \rho : \textit{Row} ⊢ M_R : B[\alpha \coloneqq \texttt{Record}^*\ \rho] \\
        Γ, \gamma : \textit{Type} ⊢ M_T : B[\alpha \coloneqq \texttt{Trace}^*\ \gamma]
      }
      { Γ ⊢ \kw{typecase}^{\alpha.B}~C~\kw{of}\ (M_B, M_I, M_S, \beta.M_L, \rho.M_R, \gamma.M_T) : B[\alpha \coloneqq C] } \]

    Reduction rules:
    \begin{itemize}
    \item $\kw{typecase}~\texttt{Bool}^*~\kw{of}\ (M_B, M_I, M_S, \beta.M_L, \rho.M_R, \gamma.M_T) \leadsto M_B$

      Need to show that $M_B:B[\alpha \coloneqq \texttt{Bool}^*]$, which is one of our hypotheses.
    \item $\kw{typecase}~\texttt{List}^*~C~\kw{of}\ (M_B, M_I, M_S, \beta.M_L, \rho.M_R, \gamma.M_T) \leadsto M_L[\beta \coloneqq C]$

      Need to show that the result of reduction $M_L[\beta \coloneqq C]$ has type $B[\alpha \coloneqq \texttt{List}^*\ C]$, the same as the typing rule.

      \[
        \infer { }
        { \Gamma \vdash M_L[\beta \coloneqq C] : B[\alpha \coloneqq \texttt{List}^*\ C] }
      \]

      Instantiating the constructor substitution lemma (Lemma \ref{lem:substitution-lemmas}) gives us
      \[ \Gamma[\beta \coloneqq C] \vdash M_L[\beta \coloneqq C] : (B[\alpha \coloneqq \texttt{List}\ \beta])[\beta \coloneqq C]
      \]
      from $\Gamma, \alpha: \mathit{Type} \vdash B: \mathit{Type}$ and $\beta \notin \textit{Dom}(Γ)$ we know that neither $B$ nor $\Gamma$ can contain $\beta$.
      Thus the only substitution for $\beta$ we need to perform is in the substitution for $\alpha$ and we can reassociate substitution like this:
      \[ \Gamma \vdash M_L[\beta \coloneqq C] : B([\alpha \coloneqq \texttt{List}\ \beta][\beta \coloneqq C])
      \]
      which is the same as
      \[ \Gamma \vdash M_L[\beta \coloneqq C] : B[\alpha \coloneqq \texttt{List}\ C] \]
    \end{itemize}
    The other cases are analogous.

  \item Case \lstinline|tracecase|: Typing rule: {\footnotesize \[    \infer
        { Γ ⊢ M : \texttt{Trace}\ A \\
          Γ, x_L : A ⊢ M_L : B \\
          Γ, x_I : \langle\mathsf{cond}: \texttt{Trace}\ \texttt{Bool}, \mathsf{then}: \texttt{Trace}\ A\rangle ⊢ M_I : B \\
          Γ, \alpha_F: \textit{Type}, x_F : \langle\mathsf{in}: \texttt{T}(\texttt{TRACE}\ \alpha_F), \mathsf{out}: \texttt{Trace}\ A\rangle ⊢ M_F : B \\
          Γ, x_C : \langle\mathsf{table}: \texttt{String}, \mathsf{column}: \texttt{String}, \mathsf{row}: \intty, \mathsf{data} : A\rangle ⊢ M_C : B \\
          Γ, \alpha_E : \textit{Type}, x_E : \langle\mathsf{left}: \texttt{T}(\texttt{TRACE}\ \alpha_E), \mathsf{right}: \texttt{T}(\texttt{TRACE}\ \alpha_E)\rangle ⊢ M_E : B \\
          Γ, x_P : \langle\mathsf{left}: \texttt{Trace}\ \intty, \mathsf{right}: \texttt{Trace}\ \intty\rangle ⊢ M_P : B \\
        }
        { Γ ⊢ \kw{tracecase}~M~\kw{of}\ (x_L.M_L, x_I.M_I, \alpha_F.x_F.M_F, x_C.M_C, \alpha_E.x_E.M_E, x_P.M_P) : B }
      \]}

    Reductions:
    \begin{itemize}
    \item $ \kw{tracecase}~\texttt{For}~C~M~\kw{of}\ (x.M_L, x.M_I, \alpha.x.M_F, x.M_C, \alpha.x.M_E, x.M_P) \leadsto M_F[\alpha \coloneqq C, x \coloneqq M] $

      We need to show
      $\infer{\star}{\Gamma \vdash M_F[\alpha \coloneqq C, x \coloneqq M] : A}$

      $\star$: We only need $M: \langle \mathsf{in}: \dots \rangle$ and $C: \mathit{Type}$, which we get by inversion of the typing rule for \lstinline|For| and the substitution lemmas.
    \end{itemize}
    The other cases are analogous.
    \qedhere
  \end{itemize}
\end{proof}

\subsection{Proof of Lemma~\ref{lem:progress-constructor}}\label{sec:prf:progress-constructor}
\begin{proof}
  By induction on the kinding derivation of $C$ or $S$ (see Figure \ref{fig:app:constructor-formation}).
  \begin{itemize}
  \item Base types \texttt{Bool}, \intty, \texttt{String} are in normal form.
  \item Type variables $\alpha$ are in normal form.
  \item Type-level functions $\lambda \alpha.C$: by IH, either $C \leadsto C'$, in which case $\lambda \alpha.C \leadsto \lambda \alpha.C'$, or $C$ is in normal form already, in which case $\lambda \alpha.C$ is in normal form, too.
  \item Type-level application $C\ D$: by IH either $C$ or $D$ may reduce, in which case the whole application reduces. Otherwise, $C$ and $D$ are in normal form.
    The following cases of $C$ do not apply, because they are ill-kinded: base types, lists, records, and traces.
    If $C$ is a normal form and a variable, application, or typerec then $C$ is a neutral form and $D$ is a normal form so $C\ D$ is a neutral (and normal) form.
    Finally, if $C$ is a type-level function, the application $\beta$-reduces.
  \item List types: by IH either the argument reduces, or is in normal form already.
  \item Record types: by IH either the argument (a row) reduces, or is in normal form already.
  \item Trace types: by IH either the argument reduces, or is in normal form already.

  \item $\kw{Typerec}\ C\ \kw{of} (C_B, C_I, C_S, \alpha.C_L, \rho.C_R, \alpha.C_T)$: by IH, either $C \leadsto C'$, in which case \lstinline|Typerec| reduces with a congruence rule, or $C$ is in one of the following normal forms:
    \begin{itemize}
    \item If $C$ is a base, list, record, or trace constructor, the \lstinline|Typerec| expression $\beta$-reduces to the respective branch.
    \item $C$ cannot be a type-level function, that would be ill-kinded.
    \item If $C$ is one of the following neutral forms: variables, applications, and \kw{Typerec}, then by IH the branches $C_B$, $C_I$, etc. either reduce and a congruence rule applies, or they are all in normal form and $\kw{Typerec}\ C\ \kw{of} (C_B, C_I, C_S, \alpha.C_L, \rho.C_R, \alpha.C_T)$ is in normal form.
    \end{itemize}


  \item The empty row $\cdot$ is in normal form.
  \item Row extensions $l:C;S$: by IH applied to $C$ and $S$ we have three cases:
    \begin{itemize}
    \item If $C \leadsto C'$, then $l:C;S \leadsto l:C';S$.
    \item If $S \leadsto S'$, then $l:C;S \leadsto l:C;S'$.
    \item If $C$ and $S$ are in normal form, then $l:C;S$ is in normal form.
    \end{itemize}

  \item $\kw{Rmap}\ C\ S$: we apply the induction hypothesis to $S$ and $C$. If either $C$ or $S$ takes a step, the whole row map expression takes a step via the respective congruence rule.
    Otherwise $S$ is in one of the following normal forms:
    \begin{itemize}
    \item Case empty row: $\kw{Rmap}\ C\ \cdot \leadsto \cdot$
    \item Case $l:D;S'$: $\kw{Rmap}\ C\ (l:D;S') \leadsto (l:C\ D; \kw{Rmap}\ C\ S')$.
    \item Case $\kw{Rmap}\ D\ U$: $\kw{Rmap}\ C\ (\kw{Rmap}\ D\ U)$ is in normal form.
    \item Case $\rho$: $\kw{Rmap}\ C\ \rho$ is in normal form.
    \end{itemize}
  \item The row variable $\rho$ is in normal form.
    \qedhere
  \end{itemize}
\end{proof}

\subsection{Proof of Lemma~\ref{lem:progress}}\label{sec:prf:progress}
\begin{proof} By induction on the typing derivation of $M$.

\begin{itemize}
\item Constants: in normal form.
\item Term variables: in normal form.
\item Term function: apply IH to body and either reduce or in normal form.
\item Fixpoint: we can always take a step by unrolling once.

\item Term application $M\ N$: apply induction hypothesis to $M$. If $M$ reduces to $M'$, then $M\ N$ reduces to $M'\ N$. Otherwise, $M$ is in \TLinks normal form. It cannot be any of the following, because these would be ill-typed: constants, type abstraction, operators, record introduction forms including record map, list introduction forms, trace introduction forms.
  In the following cases, we apply the induction hypothesis to $N$ and either reduce to $M\ N'$ or are in normal form already: variable, application, type application, record fold, tracecase, typecase.
  This leaves the following cases:
  \begin{itemize}
  \item If $M$ is a function, we $\beta$-reduce.
  \item If $M$ is of the form if-then-else, we reduce using a commuting conversion.
  \end{itemize}

\item Term-level type abstraction $\forall \alpha:M$: by IH, either $M \leadsto M'$, in which case $\forall \alpha:M \leadsto \forall \alpha:M'$, or $M$ is in normal form, in which case $\forall \alpha:M$ is in normal form as well.

\item Term-level type application $M\ C$: apply induction hypothesis to $M$. If $M$ reduces to $M'$, then $M\ C$ reduces to $M'\ C$. Otherwise, $M$ is in \TLinks normal form.
  It cannot be any of the following, because these would be ill-typed: constants, functions, operators, record introduction forms including record map, list introduction forms, trace introduction forms.
  In the following cases, the application is already in normal form: variable, application, type application, projection, record fold, tracecase, typecase.
  This leaves the following cases:
  \begin{itemize}
  \item If it is a term-level type abstraction, we $\beta$-reduce.
  \item If it is of the form if-then-else, we perform a commuting conversion.
  \end{itemize}

\item Case $\kw{if}~L~\kw{then}~M~\kw{else}~N$: apply induction hypothesis to all subterms.
  If any of the subterms reduce, then the whole if-then-else reduces.
  Otherwise, $L, M, N$ are in \TLinks normal form.
  The condition cannot be any of the following, because these would be ill-typed: functions, type abstractions, arithmetic operators, record introduction forms including record map, list introduction forms, trace introduction forms.
  In the following cases, the condition already matches the normal form: variable, application, type application, projection, record fold, tracecase, and typecase.
  This leaves the following cases for the condition:
  \begin{itemize}
  \item Constants: \lstinline|true| and \lstinline|false| reduce, other constants are ill-typed.
  \item If the condition is of the form if-then-else itself, we apply a commuting conversion.
  \item Operators with Boolean result like \lstinline!==! are in normal form.
  \end{itemize}

\item Records $\langle l = M; N \rangle$: apply induction hypothesis to $M$ and $N$. If either reduces, the whole record reduces, otherwise it is in normal form.

\item Projection $M.l$: apply induction hypothesis to $M$. If $M$ reduces to $M'$, then $M.l$ reduces to $M'.l$. Otherwise, $M$ is in \TLinks normal form.
  It cannot be any of the following, because these would be ill-typed: constants, functions, type abstraction, operators, list introduction forms, trace introduction forms.
  In any of the following cases of $M$, $M.l$ is already in normal form: variable, application, type application, projection, record map, record fold, typecase, tracecase.
  This leaves the following cases for $M$:
  \begin{itemize}
  \item If it is of the form if-then-else itself, we apply a commuting conversion.
  \item It cannot be an empty record, or a record expression where label $l$ does not appear---these would be ill-typed. If $M$ is a record literal that maps $l$ to $M'$ then $\langle l = M'; N \rangle.l$ reduces to $M'$.
  \end{itemize}

\item Record map $\kw{rmap}^S\ M\ N$: by Lemma \ref{lem:progress-constructor} we have that either $S$ reduces to $S'$, in which case $\kw{rmap}^S\ M\ N$ reduces to $\kw{rmap}^{S'}\ M\ N$, or is in normal form. Similarly, $M$ and $N$ may reduce by IH. Otherwise, we have $S$, $M$, and $N$ in normal form. By cases of $S$:
  \begin{itemize}
  \item If it is a closed row, we apply the $\beta$-rule.
  \item If it is an open row $U$, $\kw{rmap}^U\ M\ N$ is in normal form.
  \end{itemize}

\item Record fold $\kw{rfold}^S\ L\ M\ N$: same as record map.
\item Empty list: in normal form.
\item Singleton list: apply IH to element and reduce or is in normal form.
\item List concatenation: apply IH to both sides. If either reduces, the whole concatenation reduces, otherwise it is in normal form.

\item Comprehension $\kw{for}~(x \leftarrow M)~N$: apply induction hypothesis to $M$. If $M$ reduces to $M'$ then $\kw{for}~(x \leftarrow M)~N$ reduces to $\kw{for}~(x \leftarrow M')~N$. Otherwise, $M$ is in \TLinks normal form. It cannot be any of the following, because these would be ill-typed: constants, functions, type abstractions, primitive operators, record introduction forms including record map, and trace constructors. In the following cases we apply the IH to the body and either reduce or the whole comprehension is in normal form: variables, term application, type application, projection, tables, record fold, tracecase, typecase. This leaves the following cases for $M$:
  \begin{itemize}
  \item If-then-else: reduces with a commuting conversion.
  \item Empty list: the whole comprehension reduces to the empty list.
  \item Singleton list: $\beta$-reduces.
  \item List concatenation: reduces with a commuting conversion.
  \item Comprehension: reduces with a commuting conversion.
  \end{itemize}
\item Table: in normal form.
\item Trace constructors: apply IH and Lemma~\ref{lem:progress-constructor} to constituent parts. If either reduces, the whole trace constructor reduces, otherwise it is in normal form.

\item Tracecase: apply induction hypothesis to the scrutinee. If it reduces, the whole tracecase expression reduces. Otherwise it is in \TLinks normal form.
  It cannot be any of the following, because these would be ill-typed:
  constants, functions, type abstractions, primitive operators, record introduction forms, record map, empty or singleton lists, list concatenations or comprehensions, tables.
  If the scrutinee is any of the following, by IH we reduce in the branches or the whole tracecase is in normal form: variables, term application, type application, projection, record fold, tracecase, typecase.
  This leaves the following cases:
  \begin{itemize}
  \item If-then-else: reduces using commuting conversion.
  \item Trace constructor: $\beta$-reduces.
  \end{itemize}

\item Typecase: apply Lemma~\ref{lem:progress-constructor} to the scrutinee.
  Either it reduces, in which case the whole typecase expression reduces.
  Otherwise it is in normal form.
  It cannot be a type-level function, that would be ill-kinded.
  In the following cases, we apply the induction hypothesis to the branches of the typecase and reduce there, or we are in \TLinks normal form: type variables, type-level application, and typerec.
  And finally, if the outmost constructor is one of the following, a $\beta$-rule applies: bool, int, string, list, record, trace.

\item Primitive operators like \lstinline!==! and \lstinline|+|: by IH either the arguments reduce, in which case the whole expression reduces, or are in normal form, in which case the whole expression is in normal form.
  \qedhere
\end{itemize}
\end{proof}

\subsection{Proof of Lemma~\ref{lem:f-collapses}}
\begin{proof}
  By induction on the typing derivation.
  The term cannot be a record fold or typecase, because those necessarily contain a (row) type variable, which is unbound in the query context $\Gamma$.
  It cannot be a term application, type application, or tracecase, because the term in function position or the scrutinee, by IH, is of the form $x$ or $x.l$, both of which are ill-typed given that the query context $\Gamma$ does not contain function types, polymorphic types, or trace types.
  Projections $P.l$ are of the form $F.l$ or $(\kw{rmap}^U\ M\ N).l$.
  The former case reduces by IH to $x.l$ or $x.l'.l$, the first of which is okay, and the second is ill-typed.
  The latter case is impossible, because $U$ necessarily contains a row variable and would therefore be ill-typed.
  This leaves variables $x$ and projections of variables $x.l$.
\end{proof}

\subsection{Proof of Theorem~\ref{thm:nf-query-nrc}}\label{sec:prf:nf-query-nrc}
\begin{proof} By induction on the typing derivation.
  \begin{itemize}
  \item Constants, variables, empty lists, and tables are in both languages.
  \item Functions, type abstractions, and trace constructors do not have nested relational type.
  \item Function application:
    The typing rule
    \[ \infer
      { Γ ⊢ M': A \rightarrow B \\ Γ ⊢ N: A }
      { Γ ⊢ M' N : B } \]
    requires $M'$ to have a function type.
    Since $M$ is in normal form, $M'$ matches the grammar $F$.
    Lemma~\ref{lem:f-collapses} implies that $M'$ is either a variable $x$ or a projection $x.l$.
    The query context $\Gamma$ assigns record types with labels of base types to all variables --- not function types --- a contradiction.

  \item Type instantiation:
    The typing rule
    \[ \infer
      { Γ ⊢ M': \forall \alpha:K.A \\ Γ ⊢ C: K}
      { Γ ⊢ M'\ C : A[\alpha \coloneqq C]} \]
    requires $M'$ to have a polymorphic type.
    The normal form assumption requires $M'$ to match the normal form $F$.
    Therefore, Lemma~\ref{lem:f-collapses} applies, so $M'$ is either a variable $x$ or a projection $x.l$.
    The query context $\Gamma$ assigns record types with labels of base types to all variables --- a contradiction.

  \item Primitive operators, if-then-else, records, singleton list, and list concatenation: apply the induction hypothesis to the subterms.
  \item Projection $M'.l$: $M'$ is in normal form $P$, which is either of the form $F$ or a record map.
    Lemma \ref{lem:f-collapses} restricts $F$ to $x$ and $x.l'$, both of which are nested relational calculus terms.
    $P$ cannot be of the form $\kw{rmap}^U\ N'\ N''$, because $U$ necessarily contains a free type variable (see Remark \ref{remark:constructors-free-type-variable}), and thus cannot be well-typed in a query context $\Gamma$ which does not contain type variables.

  \item Record map and fold have normal forms $\kw{rmap}^U\ M'\ N$ and $\kw{rfold}^U\ L\ M'\ N$, respectively.
    $U$ necessarily contains a free type variable (see Remark \ref{remark:constructors-free-type-variable}), and thus cannot be well-typed in a query context $\Gamma$ which does not contain type variables.

  \item List comprehension $\kw{for}\ (x \leftarrow M')\ N$:
    The iteratee $M'$ is in normal form $T$, which includes tables and normal forms $F$.
    If $M'$ is a table, $x$ has closed record type with labels of base types, the induction hypothesis applies to $N$, and the whole expression is in nested relational calculus.
    If $M'$ is of the form $F$, Lemma~\ref{lem:f-collapses} applies and implies that $M'$ is either $x$ or $x.l$.
    Both cases are ill-typed, because the query context $\Gamma$ only contains variables with closed records with labels of base type --- a contradiction.

  \item Tracecase: much like the application case above, the typing derivation forces the scrutinee to be of trace type.
    The normal form forces the scrutinee to be of the form $F$, and from Lemma \ref{lem:f-collapses} follows that it has to be a variable, or projection of a variable.
    The query context $\Gamma$ assigns record types with labels of base types to all variables --- a contradiction.

  \item Typecase: the scrutinee is in normal form $E$ which contains at least one free type variable (see Remark \ref{remark:constructors-free-type-variable}).
    In a query context which only binds term variables, this cannot possibly be well-typed --- a contradiction.
    \qedhere
  \end{itemize}
\end{proof}

\end{document}